\documentclass[journal,a4paper]{IEEEtran}

\usepackage{amsmath,amssymb,amsfonts,amsthm}
\usepackage{graphicx}

\usepackage[T2A,T1]{fontenc} 
\usepackage[utf8]{inputenc}
\usepackage[english]{babel}
\usepackage[shortlabels]{enumitem}
\usepackage{bm} 

\makeatletter
\def\easycyrsymbol#1{\mathord{\mathchoice
  {\mbox{\fontsize\tf@size\z@\usefont{T2A}{cmr}{m}{n}#1}}
  {\mbox{\fontsize\tf@size\z@\usefont{T2A}{cmr}{m}{n}#1}}
  {\mbox{\fontsize\sf@size\z@\usefont{T2A}{cmr}{m}{n}#1}}
  {\mbox{\fontsize\ssf@size\z@\usefont{T2A}{cmr}{m}{n}#1}}
}}
\makeatother
\newcommand{\Ya}{\easycyrsymbol{\CYRYA}}

\usepackage{algorithm}
\usepackage[noend]{algpseudocode}

\algrenewcomment[1]{\hfill \(\triangleright\) \emph{\small #1}} 

\algnewcommand{\algorithmicand}{\textbf{and}}
\algnewcommand{\algorithmicor}{\textbf{or}}
\algnewcommand{\FOR}{\algorithmicfor}
\algnewcommand{\OR}{\algorithmicor}
\algnewcommand{\AND}{\algorithmicand}
\algnewcommand{\DO}{\algorithmicdo}

\algnewcommand{\CommentLine}[1]{\(\triangleright\) \emph{\small #1}}
\algnewcommand{\InlineFor}[2]{\algorithmicfor\ #1\ \algorithmicdo\ #2} 
\algnewcommand{\InlineIf}[2]{
  \algorithmicif\ #1\ \algorithmicthen\ #2}

\algrenewcommand\Call[2]{\nameref{#1}\ifthenelse{\equal{#2}{}}{}{\ensuremath{(#2)}}}%

\newcommand{\algoCaptionLabel}[2]{%
     \caption[\textproc{#1}]{\textproc{#1}\ifthenelse{\equal{#2}{}}{}{$(#2)$}}%
     \label{algo:#1}%
     }%

\usepackage{caption} 
\usepackage[colorlinks,linkcolor=cyan,citecolor=cyan,urlcolor=cyan]{hyperref}
\usepackage[nameinlink,capitalize,noabbrev]{cleveref}

\newtheorem{theorem}{Theorem}[section]

\newtheorem{lemma}[theorem]{Lemma}
\newtheorem{corollary}[theorem]{Corollary}
\newtheorem{definition}[theorem]{Definition}
\newtheorem{remark}[theorem]{Remark}

\numberwithin{equation}{section}

\newcommand\ifnull[3]{%
  \ifx\null#1%
    #2%
  \else%
    #3%
  \fi}

\newcommand{\ZZ}{\mathbb{Z}}
\newcommand{\field}{\mathbb{F}_q}

\newcommand{\ffield}{F}

\newcommand{\zring}{\ffield[z]} 
\newcommand{\xring}{\field[x]} 
\newcommand{\genbyuniv}[1]{\langle #1 \rangle_{\xring}} 

\newcommand{\bigO}{\mathcal{O}}
\newcommand{\softO}{\tilde \bigO}
\newcommand{\expmm}{\omega} 

\newcommand{\places}{\mathbb{P}}

\newcommand{\Pinf}{P_{\infty}}
\newcommand{\val}[1][\null]{\ifnull{#1}{v_{\Pinf}}{v_{#1}}}
\newcommand{\y}[1][\null]{\ifnull{#1}{y}{y^{(#1)}}} 
\renewcommand{\L}{\mathcal{L}} 
\newcommand{\dimL}{l}

\newcommand{\code}{\mathcal{C}_{\L}(D,G)}
\renewcommand{\r}{\vec{r}}

\newcommand{\M}{\mathcal{M}_{s,\ell,\r}} 
\newcommand{\matM}{\mat{B}_{s,\ell,\r}} 
\newcommand{\smatM}{\mat{M}_{s,\ell,\r}} 
\newcommand{\pfM}{\mat{P}_{s,\ell,\r}} 

\renewcommand{\vec}[1]{\boldsymbol{#1}}
\newcommand{\mat}[1]{\vec{#1}}
\renewcommand{\H}{\mathcal{H}} 
\newcommand{\N}{\mathcal{N}}
\newcommand{\mdim}{m} 
\newcommand{\xmatspace}[2]{\xring^{#1 \times #2}} 
\newcommand{\zmat}{\mat{0}} 
\newcommand{\idmat}{\mat{I}} 
\newcommand{\shift}{\boldsymbol{d}} 

\DeclareMathOperator{\supp}{supp}

\begin{document}

\title{Faster List Decoding of AG Codes}

\author{Peter Beelen and Vincent Neiger%
\thanks{Peter Beelen is with the Department of Applied Mathematics and Computer Science, Technical University of Denmark, 2800 Kongens Lyngby, Denmark (e-mail: pabe@dtu.dk).}%
\thanks{Vincent Neiger is with Sorbonne Université, CNRS, LIP6, F-75005 Paris, France (email: vincent.neiger@lip6.fr).}%
}

\markboth{Author version, dated March 10, 2025}{Beelen \& Neiger: Faster List Decoding of AG Codes}

\maketitle

\begin{abstract}
  In this article, we present a fast algorithm performing an instance of the
  Guruswami-Sudan list decoder for algebraic geometry codes. We show that any
  such code can be decoded in $\softO(s^2\ell^{\expmm-1}\mu^{\expmm-1}(n+g) + \ell^\expmm \mu^\expmm)$
  operations in the underlying finite field, where $n$ is the code length, $g$
  is the genus of the function field used to construct the code, $s$ is the
  multiplicity parameter, $\ell$ is the designed list size and $\mu$ is the
  smallest positive element in the Weierstrass semigroup of some chosen place.
\end{abstract}


\section{Introduction}

\emph{Context and main result.}
Algebraic geometry (AG) codes form a large class of error-correcting codes that
became famous for providing asymptotically good families of codes surpassing
the Gilbert-Varshamov bound. Such codes are constructed using algebraic curves
defined over a finite field, say $\field$ where $q$ is the cardinality of the
field. Instead of considering algebraic curves defined over $\field$, one can
also use the language of function fields with full constant field $\field$; we
follow the latter viewpoint in this article. \Cref{sec:preliminaries} gives
some more background on AG codes and function fields; for a detailed
introduction, the reader may refer to \cite{stichtenoth_algebraic_2009}.

Decoding algorithms for AG codes have been studied since the late 80's. One
important such decoder is the well-known Guruswami-Sudan (GS) list-decoder,
that can be used to decode any AG code.

Consider an AG code constructed from a function field $\ffield$ of genus $g$,
with underlying finite field $\field$ and code length $n$. The design of the GS
decoder further asks that one chooses a \emph{list size} $\ell$ and
\emph{multiplicity} parameter $s$, which are such that $s \le \ell$. It was
shown in \cite{BRS2022} that such an AG code can be decoded using a particular
instance of the GS decoder using $\softO(s\ell^{\expmm}\mu^{\expmm-1}(n+g))$
operations in $\field$, where $\mu$ is the smallest positive element in the
Weierstrass semigroup at some chosen rational place $\Pinf$ of $\ffield$. Here,
the ``soft-O'' notation $\softO(\cdot)$ is similar to the ``big-O'' notation
$\bigO(\cdot)$, but hides factors logarithmic in the parameters
\(s,\ell,\mu,n,g\). This complexity result is achieved under the mild
assumption that one has already carried out some precomputations, yielding
objects which depend on the code but not on the received word. These objects
can therefore be reused as such, with no additional computation, in each
subsequent call to the decoder (see \cite[Section~VI]{BRS2022} for more
details).

To the best of our knowledge, this complexity result is the best known one for
this decoding task in general. More precisely, any other implementation of the
GS decoder has a complexity bound which is similar or worse, with one
exception in the specific context of Reed-Solomon (RS) codes. Indeed, in this case
the above complexity from \cite{BRS2022} becomes $\softO(s\ell^{\expmm}n)$,
while there are known list decoders for these codes whose complexity is in
$\softO(s^2\ell^{\expmm-1}n)$ \cite[Section~IV]{chowdhury_faster_2015}
\cite[Sections~2.4~to~2.6]{jeannerod_computing_2017}; recall that \(s \le
\ell\).

The main goal of this paper is to refine the exploitation of efficient
univariate polynomial matrix computations in the algorithmic framework from
\cite{BRS2022}. We will show that this indeed is possible, leading to our main
result: any AG code can be list decoded, using an instance of the GS list
decoder, in complexity $\softO(s^2\ell^{\expmm-1}\mu^{\expmm-1}(n+g) +
\ell^\expmm \mu^\expmm)$. This complexity bound holds under mild assumptions
about precomputations, similar to those of \cite{BRS2022} mentioned above.
Since \(\mu \le g+1\), the term \(\ell^\expmm \mu^\expmm\) is in
\(\bigO(\ell^\expmm \mu^{\expmm-1} g)\), and therefore this new complexity
bound improves upon the best previously known bound
$\softO(s\ell^{\expmm}\mu^{\expmm-1}(n+g))$. Moreover, for Reed-Solomon codes
one has \(\mu = g+1 = 1\) and also \(\ell \le s n\), so that the new complexity
bound becomes $\softO(s^2\ell^{\expmm-1}n)$, matching the best previously known
bound in this specific case.

\smallskip
\emph{Overview of the approach.}
The GS list decoding algorithm consists of two main steps: the interpolation
step, in which one seeks a polynomial $Q(z) \in \ffield[z]$ satisfying certain
interpolation properties; and the root finding step, in which one computes
roots of the polynomial $Q(z)$. The second step is generally considered as
computationally easier than the first step.
In this paper, we keep the root finding algorithm described in
\cite[Algorithm~6]{BRS2022}, yet with a minor refinement of the complexity
analysis to ensure that it does not become the dominant step after our
improvement of the interpolation step. Specifically,
\cref{sec:decoder:rootfinding} shows a simple modification of the analysis from
\cite{BRS2022} leading to the complexity estimate
$\softO(s\ell\mu^{\expmm-1}(n+g))$, improving upon the one
$\softO(\ell^2\mu^{\expmm-1}(n+g))$ reported in \cite{BRS2022}.

We also keep the overall structure of the interpolation step \cite[Algorithm~7,
Steps 1 to 9]{BRS2022}, which performs two main tasks: first build a basis
\(\mat{B}\) of some \(\xring\)-module of interpolant polynomials, and then find
a small degree such interpolant \(Q(z)\) thanks to a suitable \(\xring\)-module
basis reduction procedure. The main novel ideas for obtaining our result are
the following:
\begin{itemize}
  \item
    For applying basis reduction to \(\mat{B}\), we rely on the algorithm from
    \cite{NeigerVu2017} for computing so-called shifted Popov forms. The same
    choice was made in \cite{BRS2022}, where it was motivated by the fact that
    this algorithm supports any shift, whereas earlier similarly efficient
    algorithms
    \cite{giorgi_complexity_2003,sarkar_normalization_2011,gupta_triangular_2012}
    focus on the unshifted case. Here, we have an additional motivation for
    this choice: a key towards our complexity improvement lies in the fact that
    the complexity of this basis reduction algorithm is sensitive to some type
    of \emph{average degree} of the input polynomial matrix.
  \item
    We describe a new algorithm to build a polynomial matrix $\mat{B}$ whose
    average degree is small, and whose rows generate all possible interpolating
    polynomials $Q(z)$. Our construction directly provides a matrix whose rows
    are \(\xring\)-linearly independent, whereas the matrix built in
    \cite{BRS2022} has redundant rows and therefore requires additional
    computations to obtain a basis of its \(\xring\)-row space, which
    furthermore typically does not have small average degree. In our case, this
    small average degree is ensured through the combination of two ingredients.
    The first one is a new description of a generating set of the module of
    interpolants (see
    \crefrange{sec:interpolant_module:generators}{sec:interpolant_module:first_Fqx_basis})
    which leads to a matrix \(\mat{B}\) with many zero entries in each row
    (like in \cite{BRS2022}), and also such that the nonzero entries are
    restricted to the first \(\mu s\) columns (unlike in \cite{BRS2022} where
    they can be found in all columns). The second ingredient is an iterative
    computation of blocks of rows of \(\mat{B}\), avoiding any degree growth at
    each stage via the computation of matrix remainders in polynomial matrix
    divisions by well-chosen matrix quotients (see
    \cref{sec:interpolant_module:small_Fqx_basis,sec:algorithms:construct_Fqx_basis}).
  \item
    A core tool in our construction of \(\mat{B}\) is a generalization of
    \cite[Algorithm~4]{BRS2022} which, given some function \(a \in \Ya(A)\),
    finds a polynomial matrix representation of the multiplication map $f \in
    \Ya(B) \mapsto af \in \Ya(A+B)$ (see \cref{sec:preliminaries} for
    definitions and notation). The version in \cite{BRS2022} was for \(A=0\),
    and we show how to generalize it to any divisor \(A\) without impacting the
    asymptotic complexity.
\end{itemize}

\smallskip
\emph{Outline.}
\cref{sec:preliminaries} presents the main definitions and preliminary results
used throughout the paper. \cref{sec:interpolant_module} focuses on bases for
the module of interpolant polynomials, starting with a versatile description of
a family of such bases, then showing a polynomial matrix representation of an
explicit choice of such a basis, and finally gathering some properties that
constructively prove the existence of a basis matrix \(\mat{B}\) with small
average degree. \cref{sec:algorithms} describes the above-mentioned
generalization of \cite[Algorithm~4]{BRS2022}, and a complete algorithm for
efficiently constructing \(\mat{B}\). Finally, \cref{sec:decoder} summarizes
the resulting list decoder and proves the announced overall complexity bound.

\smallskip
\emph{Perspectives.}
After this work, the obvious perspective is to seek further complexity
improvements beyond $\softO(s^2\ell^{\expmm-1}\mu^{\expmm-1}(n+g) + \ell^\expmm
\mu^\expmm)$.  Remark that any improvement concerning the exponents of
\(\ell\), \(s\), or \(n\) would directly imply an improvement of the
state-of-the-art complexity for the case of Reed-Solomon codes; a perhaps more
accessible target would be to reduce the dependency on the genus \(g\) or on
the quantity \(\mu\).  Although our emphasis here is on the complexity of the
decoder for a fixed code, allowing to perform some precomputations that depend
only on this code, another natural direction for further work is to carry out a
complexity analysis for these precomputations. This involves notably the
computation of Ap\'ery systems as introduced in \cite{lee_unique_2014}, which
relates directly to active research topics such as the computation of bases of
Riemann-Roch spaces (see for example \cite{Hess2002} and the literature
overview in \cite[Section~7]{BerardiniCouvreurLecerf2022}).

\section{Preliminaries}
\label{sec:preliminaries}

In this section we review some necessary concepts and notations about function
fields, AG codes, the GS list decoder, and algorithms for polynomial matrices
(i.e.\ matrices over $\xring$). We largely use the same notation as in
\cite{BRS2022} and definitions from \cite{stichtenoth_algebraic_2009}.

\subsection{Function fields and AG codes}
\label{sec:preliminaries:ff_agcodes}

Let a function field $\ffield$ of genus $g$ and full constant field $\field$ be
given. A divisor $A=\sum_i n_i A_i$ of $F$ is a formal $\ZZ$-linear combination
of places $A_i$ of $F$, such that finitely many of these places have a nonzero
coefficient. Then the support of $A$, denoted by $\supp(A)$ is the set of all
places $A_i$ of $F$ such that $n_i \neq 0$. A divisor $A$ is called effective
if for all $i$ it holds $n_i \ge 0$. This is commonly denoted by $A \ge 0$. The
degree of a place $A_i$ of $F$ is defined as the dimension of the residue field
$F_{A_i}$ of the place $A_i$, viewed as an \(\field\)-vector space. If a place
of $F$ has degree one, it is called a rational place of $F$. The degree of a
divisor $A=\sum_i n_i A_i$, is then simply defined as $\deg(A)=\sum_i n_i
\deg(A_i),$ where $\deg(A_i)$ denotes the degree of the place $A_i$.

The Riemann-Roch space of a divisor $A$ is given by
\[
  \L(A)=\{f \in \ffield\setminus\{0\} \mid (f)+A \ge 0\} \cup \{0\} ,
\]
where $(f)$ denotes the divisor of $f$.
Divisors of nonzero functions are
called principal divisors. The Riemann-Roch space $\L(A)$ is a finite
dimensional vector space over $\field$, whose dimension will be denoted by
$\dimL(A)$. The dimension of $\L(A)$ is the topic of the theorem of
Riemann-Roch \cite[Theorem~1.5.15]{stichtenoth_algebraic_2009}. In particular,
it implies that $\dimL(A) \ge \deg(A)+1-g$ and that equality holds whenever
$\deg(A) \ge 2g-1$. Moreover $\dimL(A)=0$ if $\deg(A)<0$ since the degree of a
principal divisor is zero.

Now let $P_1,\dots,P_n$ be distinct rational places of $F$ and write
$D=P_1+\cdots+P_n$. Given any divisor $G$ such that $\supp(D) \cap \supp(G) =
\emptyset$, one defines the AG code
\[
  C_{\L}(D,G)=\{(f(P_1),\dots,f(P_n)) \mid
f \in \L(G)\},
\]
where $\L(G)$ denotes the Riemann-Roch space of the divisor $G$. The dimension of the code equals
$\dimL(G)-\dimL(G-D)$, the functions in $\L(G-D)$ being
precisely all functions in $\L(G)$ that give rise to the zero codeword. In particular, $\code$ is the
zero code if $\deg(G) <0$. Moreover, the theorem of Riemann-Roch implies that
$\dim(\code)=n$, i.e.\ $\code=\field^n$, whenever $\deg(G) \ge n+2g-1$. Because
of this, we may assume $0 \le \deg(G) \le n+2g-1$ and in
particular $\deg(G) \in \bigO(n+g)$.

Further we denote by $\Pinf$ an additional rational place of $F$ not occurring
in the divisor $D$. (At first sight this seems to restrict the length of the AG
code $C_ {\L}(D,G)$, since apparently not all rational places can occur in $D$,
but as explained in \cite[Section~II.B]{BRS2022} this is not the case: if
needed, a small degree extension of the constant field will always produce
``new'' rational places from which $\Pinf$ can be chosen; a similar observation
about $\Pinf$ was made independently in \cite{Lee2022}.) We denote by $\mu \in
\ZZ_{>0}$ the smallest positive element in the Weierstrass semigroup of $\Pinf$
and by $x\in F$ a function that has pole order $\mu$ at $\Pinf$, but otherwise
is without poles. Note that \(\mu \le g+1\), since the Weierstrass semigroup
has \(g\) gaps; in fact \(\mu \le g\) as soon as the set of these gaps is not
\(\{1,\ldots,g\}\). As explained in \cite[Section~II.B]{BRS2022}, again up to a
small degree extension of the constant field if necessary, we may ensure that
\(x\) is a local parameter for a rational place \(P_0\) of \(F\) not in
\(\supp(G)\); this is is useful for the root-finding step.

For any divisor $A$ of $\ffield$, let $\Ya(A) =
\bigcup_{m=-\infty}^{\infty}\L(m\Pinf + A)$ and let $\Ya = \Ya(0)$. As in
\cite{lee_unique_2014}, for any nonzero $a \in \Ya(A)$ we denote by $\delta_A
(a)$ the smallest integer $m$ such that $a \in \L(m\Pinf + A)$, i.e.\
$\delta_A(a) = -\val(a) - \val(A)$ and let $\delta(a) = \delta_0(a) =
-\val(a)$. We will take as convention that $\delta_A(0) = - \infty$. Note that
for any $a \in \Ya(A)$ and $b \in \Ya(B)$, one has
$\delta_{A+B}(ab)=\delta_A(a)+\delta_B(b)$. We will use the quantity
$\delta_A(a)$ to indicate the ``size'' of an element $a \in \Ya(A)$; it generalizes
the degree of a univariate polynomial.  For example, the following known
result, see for example \cite[Lemma~V.3]{BRS2022} for a proof, indicates the
size of interpolating functions in $\Ya(A)$:

\begin{lemma}[{\cite[Lemma~V.3]{BRS2022}}]
  \label{lem:interp_size}
  Let $A$ be a divisor and $E = E_1 + \cdots + E_N$ for distinct rational
  places $E_1,\dots,E_N$ of $F$ different from $\Pinf$ such that $\supp(A) \cap
  \supp(E)= \emptyset$. For any $(w_1,\dots,w_N) \in \field^N$ there exists an
  $a \in \Ya(A)$ with
  \[
    \delta_A(a) \leq \deg(E) + 2g-1 - \deg(A)
  \]
  such that $a(E_j) = w_j$ for $j = 1,\dots,N$.
\end{lemma}

Since by definition, the function $x$ only has a pole in $\Pinf,$ we have $x
\in \Ya \setminus \field$. Hence, we can view $\Ya(A)$ as a free
$\xring$-module. Following
\cite{lee_unique_2014} and using the same notation as in \cite{BRS2022}, we
consider a specific kind of basis of $\Ya(A)$ as $\xring$-module, called an
Ap\'ery system of $\Ya(A)$.

\begin{definition}
  \label{def:yiA}
  For a divisor $A$ and an integer $i=0,\dots,\mu-1$, let $\y[A]_i \in \Ya(A)$
  be a function satisfying:
  \begin{enumerate}
    \item $\delta_A(\y[A]_i) \equiv i \mod \mu$,
    \item if $a \in \Ya(A)$ and $\delta_A(a) \equiv i \mod \mu$, then
      $\delta_A(\y[A]_i) \leq \delta_A(a)$.
  \end{enumerate}
  Further we define $\y_i = \y[0]_i$.
\end{definition}

Using the theorem of Riemann-Roch, it is not hard to show the following lemma,
see for example \cite[Lemma~III.3]{BRS2022} for details:

\begin{lemma}[{\cite[Lemma~III.3]{BRS2022}}]
  \label{lem:Fx_basis}
  For any divisor $A$ it holds that
  \[
    -\deg(A) \leq \delta_A(\y[A]_i) \leq 2g - 1 - \deg(A) + \mu,
  \]
  for $i = 0,\dots,\mu-1$.
\end{lemma}

As mentioned in \cite{lee_unique_2014}, $\y[A]_0,\dots,\y[A]_{\mu-1}$ is an
$\xring$-basis of $\Ya(A)$. Given $a \in \Ya(A)$, it is therefore possible
to write $a$ as an $\xring$-linear combination of these basis elements, and in
fact the \(\xring\)-coefficients of this combination are unique. As
demonstrated in \cite[Lemma~III.4]{BRS2022}, there is a very explicit upper
bound for the degree of these occurring coefficient polynomials:

\begin{lemma}[{\cite[Lemma~III.4]{BRS2022}}]
  \label{lem:deg}
  If $a = \sum_{i=0}^{\mu-1} a_i\y[A]_i \in \Ya(A)$, where $a_i \in \xring$
  and $A$ is a divisor, then
  \[
    \deg(a_i) \leq \frac{1}{\mu} \Big(\delta_A(a) - \delta_A(\y[A]_i)\Big)
              \leq \frac{1}{\mu} \Big(\delta_A(a) + \deg(A)\Big) \ .
  \]
\end{lemma}

\subsection{The Guruswami-Sudan list decoder}
\label{sec:preliminaries:gursud}

The key idea in the Guruswami-Sudan list decoding algorithm for $\code$
\cite{guruswami_improved_1999} is to find a polynomial $Q(z)=\sum_{t = 0}^\ell
z^tQ_t \in F[z]$, nonzero of degree at most $\ell$, that vanishes with
multiplicity at least $s$ at each point $(P_i,r_i)$, where $\r=(r_1,\dots,r_n)
\in \field^n$ is the received word. The idea is that if the coefficients $Q_t
\in F$ are chosen in suitable subspaces of $F$ and $\r$ has small enough
Hamming distance from the sent codeword $(f(P_1),\dots,f(P_n))$, then $Q(f)$ is
the zero element in $\ffield$ (see \cref{thm:guruswami_sudan}). This can then
be used to recover $f$ from $Q$ by finding the roots of $Q$ in \(\ffield\).

For the remainder of this paper fix $s,\ell \in \ZZ_{>0}$ with $s \le \ell$,
where $s$ is the multiplicity parameter and $\ell$ the designed list size of
the Guruswami-Sudan list decoder. The corresponding list decoding radius will
be denoted by $\tau$.

More specifically, as in \cite{BRS2022}, we restrict ourselves to the setting
where $Q = \sum_{t = 0}^\ell z^tQ_t$ with $Q_t \in \Ya(-tG)$ and define
$\delta_G(Q) = \max_t \delta_{-tG}(Q_t)$. Further, given a received word
$\r=(r_1,\dots,r_n) \in \field^n$, we define
\begin{align*}
  \M = \bigg\{ & Q = \sum_{t = 0}^{\ell}z^tQ_t \in \ffield[z] \bigm\vert  Q_t \in \Ya(-tG), \text{$Q$ has a} \\
     & \text{root of multiplicity} \ge s \text{ at $(P_j,r_j)$ for all $j$}\bigg\}.
\end{align*}
In this setting, one has the following result, which is the crux of the
correctness of the Guruswami-Sudan list decoder \cite{guruswami_improved_1999}:

\begin{theorem}[Instance of Guruswami-Sudan]
  \label{thm:guruswami_sudan}
  Let $\r\in\field^n$ be a received word and $Q \in \M$ with $\delta_G(Q) <
  s(n-\tau)$. If $f \in \L(G)$ is such that the Hamming distance between $\r$
  and $(f(P_1),\dots,f(P_n))$ is at most $\tau$, then $Q(f)=0$.
\end{theorem}

For algorithmic purposes, it is convenient to give a description of $\M$ that
is as explicit as possible. For the remainder of this article let $G_t = (t-s)D
- tG$ for $0 \le t < s$ and \(G_t = -tG\) for \(s \le t \le \ell\). From
\cite{BRS2022} we quote the following:

\begin{theorem}[{\cite[Theorem~IV.4~and~Remark~IV.7]{BRS2022}}]
  \label{thm:M-description}
  Let $R \in \Ya(G)$ be such that $R(P_j) = r_j$ for $1 \le j \le n$. Then
  \begin{align}
    \M
    & = \label{eqn:mbasis_binomial}
        \bigoplus_{t=0}^{\ell} (z-R)^t \Ya(G_t) \\
    & = \label{eqn:mbasis_binomial_shifted}
        \bigoplus_{t=0}^{s-1} (z-R)^t \Ya(G_t) \oplus \bigoplus_{t=s}^{\ell} z^{t-s} (z-R)^s \Ya(G_t).
  \end{align}
\end{theorem}

In \cite{BRS2022}, the first description (\cref{eqn:mbasis_binomial}) was used
to obtain a decoding algorithm for $\code$ with complexity
$\softO(\ell^{\expmm+1}\mu^{\expmm-1}(n+g))$, while the second description
(\cref{eqn:mbasis_binomial_shifted}) improved this to
$\softO(s\ell^{\expmm}\mu^{\expmm-1}(n+g))$. We will see in
\cref{sec:interpolant_module} that one ingredient in our improvement is to use
yet another description of \(\M\).

\subsection{Reminders on univariate polynomial matrices}
\label{sec:preliminaries:polmat}

In this paper, we will make use of a few classical notions on univariate
polynomial matrices. For brevity, since only the univariate case will be
encountered, we will just write ``polynomial matrix''.

For a polynomial matrix \(\mat{A} = [a_{i,j}]_{i,j} \in \xmatspace{\nu}{m}\),
its \emph{degree} is defined as \(\max_{i,j} \deg(a_{i,j})\) and denoted by
\(\deg(\mat{A})\); its column degree is the tuple formed by the degrees of
each of its columns. In the square case \(m=\nu\), the matrix \(\mat{A}\) is
said to be \emph{nonsingular} if \(\det(\mat{A}) \neq 0\), and
\emph{unimodular} if \(\det(\mat{A}) \in \field\setminus\{0\}\).
For a submodule $\mathcal{M} \subseteq \xring^m$ of rank $\nu$, any of its
$\xring$-bases can be seen as a set of \(\nu\) row vectors in
\(\xmatspace{1}{m}\), which can be gathered as a matrix $\mat{B}$ in
$\xmatspace{\nu}{m}$. This matrix has rank \(\nu\) and its $\xring$-row space
is \(\mathcal{M}\). In particular, if \(\mathcal{M}\) has rank \(m\), then
\(\mat{B}\) is square and nonsingular.

We are interested in finding special bases of such submodules, which have some
degree minimality property. This is realized by the so-called \emph{Popov
normal form} \cite{kailath_linear_1980}, which minimizes the maximum degree of
each basis element and further ensures a type of uniqueness. In our case, we
will need the more general \emph{shifted Popov form}
\cite{beckermann_normal_2006}, where one is allowed to put weights in the
degree measure:
 
\begin{definition}
\label{def:shifted-deg}
  For any tuple $\vec{s} = (s_1,\dots,s_m) \in \ZZ^m$ and any nonzero vector
  $\vec{v} = (v_1,\dots,v_m) \in \xring^m \setminus\{0\}$, the
  \emph{$\vec{s}$-degree} of $\vec{v}$ is defined as
\[
  \deg_{\vec{s}}(\vec{v}) = \max_k \{ \deg(v_k) + s_k \} .
\]
  If $k$ is maximal such that $\deg(v_k) + s_k = \deg_{\vec{s}}(\vec{v})$, then
  we say that $v_k$ is the \emph{$\vec{s}$-pivot} of $\vec{v}$, and $k$ is its
  \emph{$\vec{s}$-pivot index}.
\end{definition}

The tuple \(\vec{s}\) is usually referred to as a \emph{(degree) shift}.
 
\begin{definition}
  For $\vec{s} \in \ZZ^{m}$, a matrix $\mat{P} \in \xmatspace{m}{m}$ with no
  zero row is said to be in \emph{$\vec{s}$-Popov form} if all the
  $\vec{s}$-pivots of its rows lie on the diagonal, are monic, and have degrees
  strictly greater than all other entries in their respective columns.
\end{definition}

In particular, such matrices are nonsingular and can be used for representing
bases of submodules $\mathcal{M} \subseteq \xring^m$ of rank $m$. Shifted Popov
forms have good properties in that regard: for a given \(\mathcal{M}\), there
exists a unique basis of \(\mathcal{M}\) which is in \(\vec{s}\)-Popov form
\cite{beckermann_normal_2006}, and the rows of this \(\vec{s}\)-Popov basis
have minimal \(\vec{s}\)-degree (because it is in particular
\(\vec{s}\)-reduced; see \cite[Section 2.7]{zhou_fast_2012}).  Quite often, as
will be the case in this paper, one seeks the \(\vec{s}\)-Popov form $\mat{P}$
of \(\mathcal{M}\), while already knowing a nonsingular matrix \(\mat{B}\)
which forms a basis of \(\mathcal{M}\). In this context \(\mat{P}\) is the
unique \(\vec{s}\)-Popov form left-unimodularly equivalent to \(\mat{B}\), and
is called \emph{the $\vec{s}$-Popov form of $\mat{B}$}. For computing
\(\mat{P}\) from \(\mat{B}\), we will rely on \cite[Theorem~1.3]{NeigerVu2017}.

In this paper, we will need the case of rational
shifts of the form \(\vec{s} \in \frac{1}{\mu} \ZZ^m\). The above definitions extend
directly to this case, and one can still rely on the algorithms designed for
the integer case (see the proof of \cref{thm:general_small_basis}).

We will use the following result on the feasibility of polynomial matrix
division with remainder, and on the complexity of performing such divisions
using a Newton iteration-based approach.
\begin{lemma}
  \label{lem:matrix_quorem}
  Let \(\mat{A}\) and \(\mat{B}\) be matrices in \(\xmatspace{m}{m}\) with
  \(\mat{B}\) nonsingular. Then there exists a matrix \(\mat{R} \in
  \xmatspace{m}{m}\) such that \(\mat{A}-\mat{R}\) is a left multiple of
  \(\mat{B}\) and \(\deg(\mat{R}) < \deg(\mat{B})\). There is an algorithm
  \hypertarget{href_lem_matrix_quorem}{\textproc{PM-Rem}} which, on input
  \(\mat{A}\) and \(\mat{B}\), returns such a matrix \(\mat{R}\) using
  \(\softO(m^{\expmm} (\deg(\mat{A}) + \deg(\mat{B})))\) operations in
  \(\field\).
\end{lemma}
\begin{proof}
  The existence of \(\mat{R}\) such that \(\mat{A}-\mat{R}\) is a left multiple
  of \(\mat{B}\) is proved in \cite[Theorem~6.3-15,
  page~389]{kailath_linear_1980}. This reference also ensures that
  \(\mat{B}^{-1}\mat{R}\) is a so-called \emph{strictly proper} matrix
  fraction, which implies \(\deg(\mat{R}) < \deg(\mat{B})\) as showed for
  example in \cite[Lemma~6.3-10, page~383]{kailath_linear_1980}. To find
  \(\mat{R}\), one may start with computing a Popov form \(\mat{P} \in
  \xmatspace{m}{m}\) of \(\mat{B}\), which costs \(\softO(m^\expmm
  \deg(\mat{B}))\) \cite[Theorem~21]{sarkar_normalization_2011}. In particular,
  \(\mat{B}\) and \(\mat{P}\) are left-unimodularly equivalent, so that left
  multiples of \(\mat{B}\) are the same as left multiples of \(\mat{P}\).  Thus
  \(\mat{R}\) can be found as a remainder in the division of \(\mat{A}\) by
  \(\mat{P}\), since \(\deg(\mat{P}) \le \deg(\mat{B})\). Since \(\mat{P}\) is
  column reduced, to find this remainder we can apply
  \cite[Algorithm~1]{NeigerVu2017}: this boils down to one truncated expansion
  at order \(\bigO(\deg(\mat{A}))\) of the inverse of an \(m\times m\) matrix
  (whose constant term is invertible), and two multiplications of two \(m
  \times m\) matrices of degree in \(\bigO(\deg(\mat{A}) + \deg(\mat{P}))\).
  Hence the total cost is \(\softO(m^{\expmm} (\deg(\mat{A}) +
  \deg(\mat{P})))\), which concludes the proof since \(\deg(\mat{P}) \le
  \deg(\mat{B})\).
\end{proof}

\section{The interpolant module \texorpdfstring{\(\M\)}{M\_\{s,l,r\}} and polynomial matrix representations of it}
\label{sec:interpolant_module}

We now study the module $\M$ more in depth. First we generalize
\cref{thm:M-description} to get more flexibility on the choice of generators
for \(\M\), and we make such an explicit choice
(\cref{sec:interpolant_module:generators}). Then we introduce several maps and
the corresponding \(\xring\)-matrices
(\cref{sec:interpolant_module:maps_matrices}). This allows us to describe a
basis of $\M$ as an \(\xring\)-module, and to represent this module as the
$\xring$-row space of an explicit polynomial matrix $\matM$ in
\(\xmatspace{\mdim}{\mdim}\), for \(\mdim = \mu(\ell+1)\)
(\cref{sec:interpolant_module:first_Fqx_basis}). Finally, in
\cref{sec:interpolant_module:small_Fqx_basis}, we deduce another basis matrix
\(\pfM\) which is less explicit than \(\matM\), but computationally easier to
construct and manipulate. Throughout, it is assumed that the received vector is
$\r=(r_1,\dots,r_n) \in \field^n$ and that $R \in \Ya(G)$ is a function
satisfying $R(P_j) = r_j$ for $1 \le j \le n$ as well as the bound from
\cref{lem:interp_size}, $\delta_G(R) \leq n+2g-1-\deg(G)$.

\subsection{A more flexible description of \texorpdfstring{\(\M\)}{M\_\{s,l,r\}}}
\label{sec:interpolant_module:generators}

In \cref{thm:M-description}, two ways to describe the module $\M$ were given.
We now indicate a more general shape for alternative descriptions for $\M$.

\begin{theorem}
  \label{thm:bases_of_M}
  For $s \le t \le \ell$, let $f_{t}(z) = \sum_{i=0}^{t-s} f_{ti} z^{t-s-i}
  \in \ffield[z]$, where $f_{ti} \in \Ya(iG)$ and \(\deg_z(f_t) = t-s\), hence
  in particular $f_{t0} \in \field\setminus\{0\}$. Then
  \[
    \M =
    \bigoplus_{t=0}^{s-1}(z-R)^t \Ya(G_t) \oplus \bigoplus_{t=s}^{\ell} f_{t}(z) (z-R)^s \Ya(G_t).
  \]
\end{theorem}
\begin{proof}
  Using \cref{eqn:mbasis_binomial} in \cref{thm:M-description}, and since by
  definition \(G_t = -tG\) for \(t\ge s\), it is sufficient to show that
  \[
    \M = \bigoplus_{t=0}^{s-1} (z-R)^t \Ya(G_t) \oplus \bigoplus_{t=s}^{\ell} (z-R)^t \Ya(-tG)
  \]
  is equal to
  \[
    \mathcal{M} :=\bigoplus_{t=0}^{s-1}(z-R)^t \Ya(G_t) \oplus
    \bigoplus_{t=s}^{\ell}f_{t}(z)(z-R)^s \Ya(-tG).
  \]

  We first prove the inclusion ``\(\M \supseteq \mathcal{M}\)''. For this we
  show, for \(t = s,\ldots,\ell\), the inclusion $f_{t}(z)(z-R)^s \Ya(-tG)
  \subseteq \bigoplus_{k=s}^{t}(z-R)^k \Ya(-kG)$. This follows from \(f_t(z)
  \in \bigoplus_{k=s}^{t}(z-R)^{k-s} \Ya((t-k)G)\), which itself comes from
  using the binomial formula on \(f_t(z) = \sum_{j=0}^{t-s} (z-R+R)^{t-s-j}
  f_{tj}\). Indeed, this yields $f_{t}(z)=\sum_{k=s}^{t} (z-R)^{k-s}
  \tilde{f}_{tk}$ where, for \(s \le k\le t\), \(\tilde{f}_{tk} = \sum_{j =
  0}^{t-k} \binom{t-s-j}{k-s} R^{t-k-j} f_{tj}\) is in \(\Ya((t-k)G)\).

  Now we prove the inclusion ``\(\M \subseteq \mathcal{M}\)''. For this we
  show, for each $t=s,\dots,\ell$, the inclusion $(z-R)^t \Ya(-tG) \subseteq
  \bigoplus_{k=s}^{t} f_k(z)(z-R)^s \Ya(-kG)$. Similarly to the above, it is
  enough to prove $(z-R)^{t-s} \in \bigoplus_{k=s}^{t} f_k(z) \Ya((t-k)G)$. We
  proceed by induction on \(t \in \{s,\ldots,\ell\}\), showing that
  $(z-R)^{t-s} = \sum_{k=s}^{t} f_k(z) \alpha_{tk}$ for certain \(\alpha_{tk}
  \in \Ya((t-k)G)\). The property is obvious for \(t=s\), since \(f_s(z) =
  f_{s0} \in \field\setminus\{0\}\). Let \(t \in \{s+1,\ldots,\ell\}\) and
  assume the property holds from \(s\) to \(t-1\). Then
  \begin{align*}
    (z-R)^{t-s} & = \frac{1}{\tilde{f}_{tt}} \left(f_t(z) - \sum_{j=s}^{t-1} (z-R)^{j-s} \tilde{f}_{tj}\right) \\
                & = \frac{1}{\tilde{f}_{tt}} f_t(z) + \sum_{k=s}^{t-1} f_k(z) \left(\frac{-1}{\tilde{f}_{tt}} \sum_{j=k}^{t-1} \alpha_{jk} \tilde{f}_{tj}\right)
  \end{align*}
  which proves the property for \(t\), since \(\tilde{f}_{tt} \in
  \field\setminus\{0\}\) and \(\alpha_{jk} \tilde{f}_{tj} \in \Ya((t-k)G)\).
\end{proof}

As a first observation, the second description in \cref{thm:M-description} is
now an easy consequence of \cref{thm:bases_of_M} (which we proved using only
the first description in \cref{thm:M-description}). To obtain a faster decoder,
we will start from the following description of the interpolant module.

\begin{corollary}
  \label{cor:new_description}
  For $s \le t \le \ell$, let
  \[
    g_t(z)=\sum_{i=0}^{t-s} \binom{i+s-1}{i} R^i
    z^{t-s-i} \in \zring,
  \]
  then
  \[
    \M=\bigoplus_{t=0}^{s-1}(z-R)^t \Ya(G_t) \oplus
    \bigoplus_{t=s}^{\ell}g_t(z)(z-R)^s \Ya(G_t).
  \]
\end{corollary}
\begin{proof}
  We only need to check that \(f_t(z) = g_t(z)\) is a valid choice in
  \cref{thm:bases_of_M}. The first condition \(\binom{i+s-1}{i} R^i \in
  \Ya(iG)\) follows from \(R \in \Ya(G)\). The second condition \(\deg_z(g_t(z))
  = t-s\) is obvious.
\end{proof}

In this description, the polynomial \(g_t(z)(z-R)^s\) has at most \(s+1\)
nonzero coefficients, as we will see in the next lemma (\cref{lem:g_t_zeroes})
which gives an explicit formula for these coefficients.  In fact, looking back
at the description in \cref{eqn:mbasis_binomial_shifted}, the polynomial
\(z^{t-s}(z-R)^s\) also has at most \(s+1\) nonzero coefficients. Yet, the
advantage of \(g_t(z)(z-R)^s\) over \(z^{t-s}(z-R)^s\) is the range of
monomials that may appear with nonzero coefficients. Apart from the common
leading term \(z^t\), for the latter these monomials are
\(z^{t-s},z^{t-s+1},\ldots,z^{t-1}\), whereas for the former they are
\(1,z,\ldots,z^{s-1}\) independently of \(t\). As we will see in
\cref{sec:interpolant_module:first_Fqx_basis,sec:interpolant_module:small_Fqx_basis},
this particular location of nonzero coefficients is instrumental in our
approach for building an \(\xring\)-basis of \(\M\) which has small average
column degree.

\begin{lemma}
  \label{lem:g_t_zeroes}
  For $s \le t \le \ell$, let 
  \(g_t(z)\) as in \cref{cor:new_description}.
  Then $z^{t}-g_t(z)(z-R)^s$ has degree at most $s-1$. Moreover,
  for $0 \le j<s$, the coefficient of $z^j$ in $g_t(z)(z-R)^s$ equals
  \(\gamma_{t,j} R^{t-j}\), where
  \begin{equation}
    \label{eqn:dfn_gamma}
    \gamma_{t,j} = \sum_{i=s-j}^{s} \binom{s}{i} \binom{t-j-i+s-1}{s-1} (-1)^i.
  \end{equation}
\end{lemma}
\begin{proof}
  Using a classical power series expansion formula in $z^{-1}$, we obtain that
  \begin{align*}
    z^t & = (z-R)^s z^{t-s} \left( 1 - \frac{R}{z} \right)^{-s} \\
        & = (z-R)^s \sum_{i\ge 0} \binom{i+s-1}{i} R^i z^{t-s-i}.
  \end{align*}
  This shows that
  \[
  z^{t}-g_t(z)(z-R)^s= z^t - (z-R)^s \sum_{i\ge t-s+1} \textstyle\binom{i+s-1}{i} R^i z^{t-s-i}.
  \]
  Hence the polynomial $z^{t}-g_t(z)(z-R)^s$ has degree at most $s-1$.
  To prove the second part of the lemma, one can simply expand the product \(g_t(z) (z-R)^s\),
  yielding that for $0 \le j<s$  the coefficient of $z^j$ in \(g_t(z) (z-R)^s\) equals
  $\gamma_{t,j} R^{t-j}$, just as indicated.
\end{proof}

\subsection{Inclusion and multiplication maps, and their matrices}
\label{sec:interpolant_module:maps_matrices}

In this subsection we study two types of $\xring$-module homomorphisms: the
first type are inclusion maps of submodules in a module, while the second type
are maps of multiplication by some $R \in \Ya(G)$.

We will also consider matrices over \(\xring\) which represent these maps, as
this will help us describe and compute bases of \(\M\) as an \(\xring\)-module.
Note that, if the ranks as \(\xring\)-modules of the domain and codomain of the
considered map are the same, then this map can, after choosing bases, be
represented by a square $\xring$-matrix.

We start with the maps derived from the inclusions $\Ya(G_t) \subseteq
\Ya(-tG)$ for $0 \le t < s$; these inclusions follow from $G_t=(t-s)D - tG \le
-tG$.

\begin{definition}
  \label{def:inclusion_map}
  For $0 \le t < s$, the map $\imath_t: \Ya(G_t) \to \Ya(-tG)$ is defined as the
  natural inclusion map of $\Ya(G_t)$ in $\Ya(-tG)$. We denote by $\mat{D}_t
  \in \xmatspace{\mu}{\mu}$ the matrix of $\imath_t$ with respect to the ordered
  $\xring$-bases $(\y[G_t]_0,\dots,\y[G_t]_{\mu-1})$ for $\Ya(G_t)$ and
  $(\y[-tG]_0,\dots,\y[-tG]_{\mu-1})$ for $\Ya(-tG)$.
\end{definition}

\begin{remark}
  In this paper, such matrices of maps are considered in a row-wise manner. For
  example, in this definition, the \(i\)th \emph{row} of \(\mat{D}_t\) yields
  the expression of \(\imath_t(\y[G_t]_i)\) as an \(\xring\)-linear combination
  of the mentioned basis of \(\Ya(-tG)\).
\end{remark}

For deriving complexity estimates, we will use the
following bound on the degree of any single entry of \(\mat{D}_t\).

\begin{lemma}
  \label{lem:deg_D_t}
  The matrix \(\mat{D}_t\) is nonsingular and \(\deg(\mat{D}_t)\) is in
  $\bigO(s(n+g)/\mu)$.
\end{lemma}
\begin{proof}
  Since \(\imath_t\) is an injection, \(\mat{D}_t\) is nonsingular. Let
  \([p_{ij}]_{0 \le i,j < \mu}\) be the entries of the matrix \(\mat{D}_t\), so
  that $\y[G_t]_i=\sum_{j=0}^{\mu-1} p_{ij} \y[-tG]_j$ for \(0\le i<\mu\) and
  \(\deg(\mat{D}_t) = \max_{ij} \deg(p_{ij})\). Using
  \cref{lem:Fx_basis,lem:deg} we see that
  \begin{align*}
    \deg(p_{ij}) & \le \frac{1}{\mu}(\delta_{-tG}(\y[G_t]_i) -t\deg(G)) \\
                 & \le \frac{1}{\mu}(\delta_{G_t}(\y[G_t]_i) -t\deg(G))\\
                 & \le \frac{1}{\mu}(2g-1+\mu-\deg(G_t) -t\deg(G)) \\
                 & =\frac{1}{\mu}(2g-1+\mu+(s-t)n) \;\;\in \bigO(s(n+g)/\mu).
                 \qedhere
  \end{align*}
\end{proof}

We will also use the following, similarly defined inclusion maps and matrices,
where we have defined \(H_t = -sD-tG\) for \(0 \le t < s\).

\begin{definition}
  \label{def:inclusion_map_superset}
  For $0 \le t < s$, the map $\jmath_t: \Ya(H_t) \to \Ya(-tG)$ is defined
  as the natural inclusion map of $\Ya(H_t)$ in $\Ya(-tG)$. We denote by
  $\mat{E}_t \in \xmatspace{\mu}{\mu}$ the matrix of $\jmath_t$ with
  respect to the $\xring$-bases $(\y[H_t]_0,\dots,\y[H_t]_{\mu-1})$
  for $\Ya(H_t)$ and $(\y[-tG]_0,\dots,\y[-tG]_{\mu-1})$ for $\Ya(-tG)$.
\end{definition}

This matrix \(\mat{E}_t\) satisfies properties similar to those of
\(\mat{D}_t\).

\begin{lemma}
  \label{lem:deg_E_t}
  The matrix \(\mat{E}_t\) is nonsingular and \(\deg(\mat{E}_t)\) is in
  $\bigO(s(n+g)/\mu)$.
\end{lemma}
\begin{proof}
  The proof can be directly adapted from that of \cref{lem:deg_D_t}.
\end{proof}

We now turn our attention to the maps of multiplication by some $R \in \Ya(G)$.

\begin{definition}
  \label{def:R_t}
  For $1 \le t \le \ell$, we let the multiplication map $R_t: \Ya(-tG) \to
  \Ya(-(t-1)G)$ be defined by $R_t : f \mapsto Rf$.  We denote by $\mat{R}_t
  \in \xmatspace{\mu}{\mu}$ the matrix of $R_t$ with respect to the ordered
  $\xring$-bases $(\y[-tG]_0,\dots,\y[-tG]_{\mu-1})$ for $\Ya(-tG)$ and
  $(\y[-(t-1)G]_0,\dots,\y[-(t-1)G]_{\mu-1})$ for $\Ya(-(t-1)G)$.
\end{definition}

Although this definition is valid for any \(R \in \Ya(G)\), recall that here we
consider specifically \(R\) such that \(\delta_G(R) \leq n+2g-1-\deg(G)\). This
allows us to bound the degree of any single entry of \(\mat{R}_t\), as follows.

\begin{lemma}
  \label{lem:deg_R_t}
  The matrix degree of $\mat{R}_t$ is in $\bigO((n+g)/\mu)$.
\end{lemma}
\begin{proof}
  Using \cref{lem:Fx_basis}, we see that for any $0 \le i,t < \mu$,
  \begin{align*}
    \delta_{-(t-1)G}(R\y[-tG]_i) & = \delta_G(R) + \delta_{-tG}(\y[-tG]_i) \\
                                 & \le \delta_G(R) + 2g-1 + t\deg(G) + \mu.
  \end{align*}
  Let \([p_{ij}]_{0 \le i,j < \mu}\) be the entries of \(\mat{R}_t\), so that
  $R\y[-tG]_i = \sum_{j=0}^{\mu-1} p_{ij} \y[-(t-1)G]_j$ for \(0\le i<\mu\) and
  \(\deg(\mat{R}_t) = \max_{ij} \deg(p_{ij})\). Then from \cref{lem:deg} and
  the fact that $\delta_G(R) \leq n+2g-1-\deg(G)$, we obtain
  \begin{align*}
    & \deg(p_{ij}) \le \frac{1}{\mu}\left(\delta_{-(t-1)G}(R\y[-tG]_i) + \deg(-(t-1)G)\right) \\
                 & \le \frac{1}{\mu}\left(\delta_G(R) + 2g-1 + t\deg(G) + \mu - (t-1) \deg(G)\right) \\
                 & = 1 + \frac{1}{\mu}\left(\delta_G(R) + 2g-1 + \deg(G)\right) \\
                 & \le 1 + \frac{1}{\mu}\left(n + 4g-2\right).
    \qedhere
  \end{align*}
\end{proof}

In what follows we will also use this notation:
\begin{definition}
  \label{def:R_tj}
  For \(0\le t \le \ell\) and \(0 \le j \le \ell\), the matrix
  \(\mat{R}^{(t,j)} \in \xmatspace{\mu}{\mu}\) is defined as
  \[
    \mat{R}^{(t,j)} =
    \left\{
      \begin{array}{l}
        \mat{R}_t \mat{R}_{t-1} \cdots \mat{R}_{j+1} \text{ for } 0 \le j < t; \\
        \text{the } \mu \times \mu \text{ identity matrix } \idmat\text{ for } j=t; \\
        \text{the } \mu \times \mu \text{ zero matrix } \zmat \text{ for } t < j \le \ell.
      \end{array}
    \right.
  \]
\end{definition}

The definition of the \(\mat{R}_t\)'s implies that, for \(0 \le j \le t\),
\(\mat{R}^{(t,j)}\) is the matrix of the following map of multiplication by
\(R^{t-j}\):
\[
  R_{j+1} \circ R_{j+2} \circ \cdots \circ R_t: \Ya(-tG) \to \Ya(-jG), f \mapsto R^{t-j} f,
\]
in the ordered bases \((\y[-tG]_0,\dots,\y[-tG]_{\mu-1})\) for \(\Ya(-tG)\) and
\((\y[-jG]_0,\dots,\y[-jG]_{\mu-1})\) for \(\Ya(-jG)\). Similarly if \(j \le t
< s\), then \(\mat{D}_t \mat{R}^{(t,j)}\) is the matrix of the map
\[
  R_{j+1} \circ R_{j+2} \circ \cdots \circ R_t \circ \imath_t: \Ya(G_t) \to \Ya(-jG), f \mapsto R^{t-j} f,
\]
in the ordered bases \((\y[G_t]_0,\dots,\y[G_t]_{\mu-1})\) for \(\Ya(G_t)\) and
\((\y[-jG]_0,\dots,\y[-jG]_{\mu-1})\) for \(\Ya(-jG)\).

\subsection{A first polynomial matrix basis of \texorpdfstring{\(\M\)}{M\_\{s,l,r\}}}
\label{sec:interpolant_module:first_Fqx_basis}

From \cref{thm:bases_of_M}, one may deduce a basis of \(\M\) as an
\(\xring\)-module.

\begin{lemma}
  \label{lem:Fqx_bases_of_M}
  \(\M\) is an \(\xring\)-module of rank \(\mdim := \mu(\ell+1)\), and admits
  the following basis:
  \begin{align*}
    & \left\{ (z-R)^t \y[G_t]_i \mid 0 \le t < s, 0 \le i < \mu \right\} \\
    & \qquad \bigcup\;\; \left\{ f_{t}(z) (z-R)^s \y[G_t]_i \mid s \le t \le \ell, 0 \le i < \mu \right\},
  \end{align*}
  for any family of polynomials \(\{f_{t}(z) \in \zring \mid s \le t \le
  \ell\}\) as in \cref{thm:bases_of_M}.
\end{lemma}
\begin{proof}
  Let \(\mathcal{B}\) be the claimed basis of \(\M\). Since
  $\genbyuniv{\y[G_t]_0,\dots,\y[G_t]_{\mu-1}} = \Ya(G_t)$ for \(0\le t \le
  \ell\), from \cref{thm:bases_of_M} it follows both that \(\mathcal{B}
  \subseteq \M\) and that any element of \(\M\) is an \(\xring\)-linear
  combination of polynomials in \(\mathcal{B}\); whence
  \(\genbyuniv{\mathcal{B}} = \M\). To prove that \(\mathcal{B}\) is a basis,
  it remains to show that its elements are \(\xring\)-linearly independent. Let
  \((\alpha_{t,i})_{0\le i < \mu, 0 \le t < k} \in \xring^\mdim\) be a tuple
  such that
  \begin{align*}
    & \sum_{0 \le t < s, 0 \le i < \mu} \alpha_{t,i} (z-R)^t \y[G_t]_i \\
    & +
    \sum_{s \le t \le \ell, 0 \le i < \mu} \alpha_{t,i} f_{t}(z) (z-R)^s \y[G_t]_i
    =
    0.
  \end{align*}
  Since \(f_t(z)(z-R)^s\) has degree \(t\), the polynomials
  \[
    \{(z-R)^t \mid 0 \le t < s\} \cup \{f_t(z) (z-R)^s \mid s \le t \le \ell\}
  \]
  form a basis of
  the \(\ffield\)-vector space \(\zring_{\deg_z \le \ell}\). Thus, from the
  above identity we deduce that \(\sum_{0\le i< \mu} \alpha_{t,i} \y[G_t]_i =
  0\) for \(0 \le t \le \ell\). By definition of the \(\y[G_t]_i\)'s, this
  implies \(\alpha_{t,i} = 0\) for all \(t\) and \(i\).  Hence the rows of \(\mathcal{B}\)
  form a basis of \(\M\), and the rank of \(\M\) is the cardinality \(\mdim\) of
  \(\mathcal{B}\).
\end{proof}

To represent such a basis of \(\M\) as a matrix \(\matM\) over \(\xring\), we
see \(\M\) as a submodule of the free \(\xring\)-module \(\bigoplus_{0 \le t
\le \ell} z^t \Ya(-tG)\) of rank \(\mdim\), with basis \(z^j \y[-tG]_k, 0\le t
\le \ell, 0 \le k < \mu\). The following \(\xring\)-module isomorphism will be
useful for describing \(\matM\):
\[
    \varphi_\ell: \xmatspace{1}{\mdim} \;\; \to \;\; \bigoplus_{t=0}^\ell z^t \Ya(-tG),
\]
which maps
\([p_{0,0} \cdots p_{0,\mu-1} \;|\; \cdots \;|\; p_{\ell,0} \cdots p_{\ell,\mu-1}]\)
to
\(\sum_{t=0}^\ell \sum_{k=0}^{\mu-1} p_{t,k} \y[-tG]_k z^t\).
Then, the rows of \(\matM\) are the preimages by \(\varphi_\ell\) of the
elements of the basis of \(\M\) described in \cref{lem:Fqx_bases_of_M}.
Choosing specifically for \(f_t(z)\) the polynomial \(g_t(z)\) described in
\cref{sec:interpolant_module:generators}, and using the maps and matrices
defined in \cref{sec:interpolant_module:maps_matrices}, we obtain the following
explicit description of \(\matM\).

\begin{definition}
  \label{def:Fqx_basis_matrix_of_M}
  Let \(\mdim = \mu(\ell+1)\).
  The matrix $\matM \in \xmatspace{\mdim}{\mdim}$ is defined by
  blocks as $\matM = [\begin{smallmatrix} \mat{D} & \zmat \\ \mat{R}
  & \idmat \end{smallmatrix}]$ where
  \begin{itemize}
    \item \(\idmat\) is the \((\mdim-\mu s) \times (\mdim - \mu s)\) identity
      matrix;
    \item \(\zmat\) is the \((\mu s) \times (\mdim-\mu s)\) zero matrix;
    \item \(\mat{D} \in \xmatspace{(\mu s)}{(\mu s)}\) is defined by blocks as
      \begin{align*}
        \mat{D}
        & =
        \begin{bmatrix}
          \gamma_{t,j} \mat{D}_t \mat{R}^{(t,j)}
        \end{bmatrix}_{0\le t < s, 0 \le j < s}
        \\
        & =
        \begin{bmatrix}
          \mat{D}_0 \\
          -\mat{D}_1 \mat{R}_1 & \mat{D}_1 \\
          \mat{D}_2 \mat{R}_2 \mat{R}_1 & -2 \mat{D}_2 \mat{R}_2 & \mat{D}_2 \\
          \vdots &  &  & \ddots
        \end{bmatrix},
      \end{align*}
      where \(\gamma_{t,j} = (-1)^{t-j} \binom{t}{j}\) for \(0 \le t <s\), \(0
      \le j < s\);
    \item \(\mat{R} \in \xmatspace{(\mdim - \mu s)}{(\mu s)}\) is defined by
      blocks as
      \begin{align*}
        \mat{R}
        & =
        \begin{bmatrix}
          \gamma_{t,j} \mat{R}^{(t,j)}
        \end{bmatrix}_{s \le t \le \ell, 0 \le j < s}
        \\
        & =
        \begin{bmatrix}
          \gamma_{s,0} \mat{R}^{(s,0)} & \cdots & \gamma_{s,s-1} \mat{R}^{(s,s-1)} \\
          \vdots & & \vdots \\
          \gamma_{\ell,0} \mat{R}^{(\ell,0)} & \cdots & \gamma_{\ell,s-1} \mat{R}^{(\ell,s-1)}
        \end{bmatrix},
      \end{align*}
      where \(\gamma_{t,j}\) for \(s \le t \le \ell, 0 \le j < s\) is defined
      in \cref{eqn:dfn_gamma}.
  \end{itemize}
\end{definition}

\begin{theorem}
  \label{thm:first_Fqx_basis_matrix_of_M}
  The matrix \(\matM\) from \cref{def:Fqx_basis_matrix_of_M} is a basis of
  \(\M\), seen as an \(\xring\)-submodule of \(\bigoplus_{0 \le t \le \ell} z^t
  \Ya(-tG)\). More precisely, indexing the rows of \(\matM\) from \(0\) to
  \(\mdim-1\), for \(0\le t \le \ell\) and \(0 \le i < \mu\) its row at index
  \(t\mu + i\) is \(\varphi_\ell^{-1}((z-R)^t \y[G_t]_i)\) if \(t < s\), and
  \(\varphi_\ell^{-1}(g_t(z) (z-R)^s \y[G_t]_i)\) if \(s \le t \le \ell\).
\end{theorem}

\begin{proof}
  This follows directly from the construction of \(\matM\) and from the
  definition of \(\mat{D}_t\) and \(\mat{R}^{(t,j)}\) in
  \cref{sec:interpolant_module:maps_matrices}. Indeed, for \(0 \le t \le \ell\)
  and \(0 \le i < \mu\), the image by \(\varphi_\ell\) of the row \(t\mu+i\) of
  \(\matM\) as built in \cref{def:Fqx_basis_matrix_of_M} is
  \[
  \sum_{j=0}^{t} (-1)^{t-j} \binom{t}{j} R^{t-j}\y[G_t]_i z^j = (z-R)^t\y[G_t]_i \text{ if } t < s,
  \]
  and
  \begin{align*}
       & z^t\y[-tG]_i + \sum_{j=0}^{s-1} \gamma_{t,j} R^{t-j}\y[-tG]_i z^j \\
                & = g_t(z) (z-R)^s \y[-tG]_i \text{ if } s \le t \le \ell,
  \end{align*}
  where the last identity comes from \cref{lem:g_t_zeroes}.
\end{proof}

Observe in particular the effect of our choice of polynomials \(g_t(z)\) from
\cref{sec:interpolant_module:generators}: only the left \(\mu s\) columns
$[\begin{smallmatrix} \mat{D} \\ \mat{R} \end{smallmatrix}]$ of \(\matM\) are
nontrivial, while the remaining columns $[\begin{smallmatrix} \zmat \\ \idmat
\end{smallmatrix}]$ are standard basis vectors. A parallel can be drawn with the remark
on the monomials appearing in \(g_t(z) (z-R)^s\), in
\cref{sec:interpolant_module:generators}. In contrast, the two descriptions
from \cite{BRS2022} recalled in \cref{thm:M-description} lead to matrices which
are block-triangular as well, but with a lower triangular part which is either
dense (if using \cref{eqn:mbasis_binomial}) or is a band matrix (if using
\cref{eqn:mbasis_binomial_shifted}). Although in the latter case the number of
nonzero blocks is the same as in \(\matM\), the fact that these nonzero blocks
are confined to the leftmost columns brings us closer to knowing an
\(\xring\)-basis of \(\M\) with small average column degree, as we are going to
see in the next subsection.

\subsection{A small-degree polynomial matrix basis of \texorpdfstring{\(\M\)}{M\_\{s,l,r\}}}
\label{sec:interpolant_module:small_Fqx_basis}

Keeping the same notation as in the previous subsection, we deduce from the matrix
\(\matM\) a whole collection of suitable matrices for representing bases of
\(\M\) as an \(\xring\)-module, which all share the property that their
rightmost \(\mdim-\mu s\) columns are the standard basis vectors
$[\begin{smallmatrix} \zmat \\ \idmat \end{smallmatrix}]$. It is within this
collection that we will find an \(\xring\)-basis of \(\M\) with small average
column degree, which means that its shifted Popov form can be computed
efficiently.

\begin{theorem}
  \label{thm:general_small_basis}
  For any matrix \(\mat{\bar{R}} \in \xmatspace{(\mdim-\mu s)}{(\mu s)}\) such
  that \(\mat{R} - \mat{\bar{R}}\) is a left multiple of \(\mat{D}\), the
  matrix \(\smatM := [\begin{smallmatrix} \mat{D} & \zmat \\ \mat{\bar{R}} &
  \idmat \end{smallmatrix}]\in \xmatspace{\mdim}{\mdim}\) is a basis of \(\M\),
  seen as an \(\xring\)-submodule of \(\bigoplus_{0 \le t \le \ell} z^t
  \Ya(-tG)\). In particular, if \(\mat{\bar{R}}\) has degree in
  \(\bigO(s(n+g)/\mu)\), then the sum of column degrees of \(\smatM\) is in
  \(\bigO(s^2(n+g))\), and the \(\shift\)-Popov form of \(\smatM\) can be
  computed in \(\softO(s^2 \ell^{\expmm-1} \mu^{\expmm-1} (n+g) + \ell^\expmm \mu^\expmm)\) operations in
  \(\field\) for any shift \(\shift \in \frac{1}{\mu} \ZZ^{(\ell+1)\mu}\).
\end{theorem}
\begin{proof}
  The matrices \(\matM\) and \(\smatM\) are left-unimodularly equivalent, since
  \[
    \begin{bmatrix}
      \idmat & \zmat \\
      -\mat{Q} & \idmat
    \end{bmatrix}
    \begin{bmatrix}
      \mat{D} & \zmat \\
      \mat{R} & \idmat
    \end{bmatrix}
    =
    \begin{bmatrix}
      \mat{D} & \zmat \\
      \mat{\bar{R}} & \idmat
    \end{bmatrix}
  \]
  where \(\mat{Q} \in \xmatspace{(\mdim - \mu s)}{(\mu s)}\) is the quotient
  matrix such that \(\mat{R} = \mat{Q} \mat{D} + \mat{\bar{R}}\). Hence
  the rows of \(\smatM\) form a basis of \(\M\).

  The degree bounds in \cref{sec:interpolant_module:maps_matrices} and the
  construction of \(\mat{R}^{(t,j)}\) show that, for \(0\le t < s\) and \(0 \le
  j < s\), both \(\deg(\mat{D}_t)\) and \(\deg(\mat{R}^{(t,j)})\) are in
  \(\bigO(s(n+g)/\mu)\). Therefore the degree of \(\mat{D} = [\gamma_{t,j}
  \mat{D}_t \mat{R}^{(t,j)}]_{0\le t < s, 0\le j < s}\) is also in
  \(\bigO(s(n+g)/\mu)\). Then, under the assumption on \(\deg(\mat{\bar{R}})\)
  stated in the theorem, the first \(\mu s\) columns of \(\smatM\) have degree in
  \(\bigO(s(n+g)/\mu)\), whereas its last \(\mdim-\mu s\) columns have degree
  \(0\). The claimed bound on the column degrees of \(\smatM\) follows.

  Computing the shifted Popov form of \(\smatM\), for any integer shift
  \(\shift \in \ZZ^{(\ell+1)\mu}\), can then be performed in
  \(\softO(\ell^\expmm\mu^\expmm \lceil \frac{s^2(n+g)}{\ell\mu} \rceil)\)
  operations in \(\field\) \cite[Theorem~1.3]{NeigerVu2017}. It has been showed
  that the case of a shift \(\shift\) with fractional entries in
  \(\frac{1}{\mu} \ZZ^{(\ell+1)\mu}\) directly reduces to the case of an
  integer shift by both permuting the matrix columns appropriately and rounding
  down the entries of \(\shift\) to integers; see
  \cite[Section\,III.A]{nielsen_sub-quadratic_2015} \cite[Theorem~V.9]{BRS2022}
  for the present case, and \cite[Section~1.3.4]{vincent_neiger_bases_2016} for
  transforming more generally any module monomial ordering on
  \(\xring^{(\ell+1) s}\) into a corresponding shift in \(\ZZ^{(\ell+1)\mu}\).
  The inequality \(\lceil \frac{s^2(n+g)}{\ell\mu} \rceil <
  \frac{s^2(n+g)}{\ell\mu} + 1\) leads to the cost bound stated in the theorem.
\end{proof}

Finally, in \cref{thm:specific_small_basis} we will make the above result more
effective by describing an explicit construction of such a small-degree matrix
\(\mat{\bar{R}}\), using remainders in the matrix division of \(\mat{R}\)
modulo the matrices \(\mat{E}_0,\ldots,\mat{E}_{s-1}\) defined in
\cref{sec:interpolant_module:maps_matrices}. Here are some explanations why
such a degree reduction is needed, and not straightforward.

\begin{remark}
  Observe that the matrix \(\mat{R}\), as defined in
  \cref{sec:interpolant_module:first_Fqx_basis}, may have degrees too large for
  our purpose; that is, simply taking \(\mat{\bar{R}}=\mat{R}\) in the above
  theorem is not interesting as \(\deg(\mat{R})\) is most likely not in
  \(\bigO(s(n+g)/\mu)\). In fact, by definition \(\mat{R}\) has \(\ell+1-s\)
  blocks of \(\mu\) rows each with \(s \mu\) columns, and the degree of the
  \(i\)th block of rows is in \(\bigO((s+i) \frac{n+g}{\mu})\). In total, the
  dense representation of \(\mat{R}\) therefore uses
  \[
    \bigO\left(\sum_{1 \le i \le \ell+1-s} \mu (\mu s) (s+i) \frac{n+g}{\mu}\right)
    \subseteq \bigO\left( \ell^2 s \mu (n+g) \right)
  \]
  coefficients from \(\field\), and this asymptotic bound can be reached.
  Indeed, it is reached already in the case of Reed-Solomon codes, where
  \(\mat{R}^{(t,j)}\) is a polynomial in \(\xring\), which is the power
  \(R^{t-j}\) of some polynomial \(R\in\xring\) whose degree is \(n-1\)
  generically. Thus, simply the size of the storage of \(\mat{R}\) can already
  be in conflict with our target complexity. This also implies that we must aim
  to compute a smaller degree \(\mat{\bar{R}}\) without computing all of
  \(\mat{R}\).
\end{remark}

\begin{theorem}
  \label{thm:specific_small_basis}
  For each \(s \le t \le \ell\) and \(0 \le j < s\), there exists a matrix
  \(\mat{\bar{R}}^{(t,j)}\) such that \(\mat{R}^{(t,j)} -
  \mat{\bar{R}}^{(t,j)}\) is a left multiple of \(\mat{E}_j\) and
  \(\deg(\mat{\bar{R}}^{(t,j)}) < \deg(\mat{E}_j)\). Then, the matrix
  \(\mat{\bar{R}} = [\gamma_{t,j} \mat{\bar{R}}^{(t,j)}]_{s \le t \le \ell, 0
  \le j < s} \in \xmatspace{(\mdim-\mu s)}{(\mu s)}\) has degree in
  \(\bigO(s(n+g)/\mu)\) and is such that \(\mat{R} - \mat{\bar{R}}\) is a left
  multiple of \(\mat{D}\).
\end{theorem}
\begin{proof}
  The existence of \(\mat{\bar{R}}^{(t,j)}\) with the specified properties
  follows directly from \cref{lem:matrix_quorem}. The bound on
  \(\deg(\mat{\bar{R}})\) follows from \(\deg(\mat{E}_j) \in
  \bigO(s(n+g) / \mu)\), proved in \cref{lem:deg_E_t}. By construction,
  \(
    \mat{R} - \mat{\bar{R}}
    =
    [
      \gamma_{t,j} (\mat{R}^{(t,j)} - \mat{\bar{R}}^{(t,j)})
    ]_{s \le t \le \ell, 0 \le j < s}
  \)
  is a left multiple of
  \[
    \mat{E} =
    \begin{bmatrix}
      \mat{E}_0 \\
      & \mat{E}_1 \\
      & & \ddots \\
      & & & \mat{E}_{s-1}
    \end{bmatrix}
    \in \xmatspace{(\mu s)}{(\mu s)},
  \]
  hence it remains to prove that \(\mat{E}\) is itself a left multiple of
  \(\mat{D}\).

  We now show that each row of \(\mat{E}\) is a left multiple of \(\mat{D}\).
  Let \(0\le i < \mu\) and \(0 \le t < s\). Similarly to the considerations in
  \cref{sec:interpolant_module:first_Fqx_basis}, we observe that the row
  \(t\mu+i\) of \(\mat{E}\) corresponds to the polynomial \(z^t \y[H_t]_i \in
  z^t \Ya(H_t)\), which is in \(\M\) since any function in \(\Ya(H_t) =
  \Ya(-sD-tG)\) has valuation at least \(s\) at each of the places \(P_1,
  \ldots, P_n\). Since \(z^t \y[H_t]_i\) has \(z\)-degree less than \(s\), it is in \(\bigoplus_{0 \le j < s} (z-R)^j
  \Ya(G_j)\) by
  \cref{thm:M-description}, which means that the row \(t\mu+i\) of \(\mat{E}\) is a left
  multiple of \(\mat{D}\).
\end{proof}

\section{Efficient construction of a polynomial matrix basis of \texorpdfstring{\(\M\)}{M\_\{s,l,r\}}}
\label{sec:algorithms}

\subsection{Computing multiplication maps}
\label{sec:algorithms:basis_products}

Consider the following problem: given two divisors $A$ and $B$ and a function
$a \in \Ya(A)$, compute the products $y_0^{(B)}a,\dots,y_{\mu-1}^{(B)}a \in
\Ya(A+B)$, expressed in the basis $y_0^{(A+B)},\dots,y_{\mu-1}^{(A+B)}$. In
order to do this, we generalize \cite[Algorithm~4]{BRS2022}. We follow the same
approach as in \cite{BRS2022} for showing the correctness and complexity of
this generalization.

\begin{definition}
  For any $Q(z) \in \ffield[z]$, for any rational place $P \in
  \places_{\ffield}$ that is not a pole of any of the coefficients of $Q(z)$,
  and for $\alpha \in \field$, we denote by $Q(P,\alpha)$ the evaluation of
  $Q(\alpha) \in \ffield$ at $P$.
\end{definition}

\begin{definition}
  \label{def:Nae}
  Let $A, B$ be divisors and let $E = E_1 + \cdots + E_N$ for distinct rational
  places $E_1,\dots,E_N$ of $F$ different from $\Pinf$ such that $\supp(A) \cap
  \supp(E)= \emptyset$ and $\supp(B) \cap \supp(E)=\emptyset$. For $a\in
  \Ya(A)$, we define the $\xring$-module
  \begin{align*}
    \N_{A,B,E}(a) & = \{ Q = Q_0 + Q_1 z \in \Ya(A+B) \oplus z\Ya(B) \\
                  &  \text{ such that } Q(P,a(P)) = 0 \text { for all } P \in \supp(E)
    \}.
  \end{align*}
\end{definition}

In the following lemmas, we use the same notation $A$, $B$, $E$ as in
\cref{def:Nae}.

\begin{lemma}
  \label{lem:fxbasis_basismem}
  Let $a \in \Ya(A)$. If $Q = Q_0 + z Q_1 \in \N_{A,B,E}(a)$ with
  \[
    \max\{
    \delta_{A+B}(Q_0),
    \delta_B(Q_1) + \delta_A(a)
    \} < \deg(E) - \deg(A+B),
  \]
  then $Q(a) = 0$, i.e.\ $Q \in (z-a)\Ya(B)$.
\end{lemma}
\begin{proof}
  Since $Q \in \N_{A,B,E}(a)$, we have $Q(a) \in \Ya(A+B)$. Hence by definition
  of $\delta_{A+B}$, we have $Q(a) \in \L(\delta_{A+B}(Q(a))\Pinf+A+B)$. Since
  for all $E_j \in \supp(E)$, we have $Q(a)(E_j)=0$ and $\supp(E) \cap
  (\supp(A) \cup \supp(B) \cup \{\Pinf\})= \emptyset$, we conclude that
  $Q(a) \in \L(\delta_{A+B}(Q(a))\Pinf+A+B-E)$.  Moreover,
\begin{align*}
    \delta_{A+B}(Q(a))
    & \;\leq\;  \max\{ \delta_{A+B}(Q_0), \delta_{A+B}(Q_1 a) \} \\
    & \;=\; \max\{ \delta_{A+B}(Q_0), \delta_B(Q_1) + \delta_A(a) \} \\
    & \;<\; \deg(E) - \deg(A+B),
\end{align*}
which ensures that the aforementioned Riemann-Roch space is trivial.
\end{proof}

Like in \cite{BRS2022}, our generalization will use the notion of an
\emph{$x$-partition} of $E$. We recall the definition, see also
\cite[Definition~V.4]{BRS2022}; the existence of an $x$-partition of $E$ was
shown in \cite[Lemma~V.6]{BRS2022}.

\begin{definition}
  If $E = E_1 + \cdots + E_N$, where $E_1,\dots,E_N$ are distinct rational
  places different from $\Pinf$, and $U_0,\dots,U_{\mu-1}$ are effective divisors
  satisfying
  \begin{enumerate}
  \item $E = U_0 + \dots + U_{\mu-1}$,
  \item $\supp(U_i) \cap \supp(U_j) = \emptyset$ for all $i \neq j$,
  \item $|\deg(U_i) - \deg(U_j)| \leq 1$ for all $i,j$,
  \item for any $E_j,E_k \in \supp(U_i)$ it holds that
    $x(E_j) = x(E_k) \Leftrightarrow E_j = E_k$,
  \end{enumerate}
  then we will say that $U_0,\dots,U_{\mu-1}$ is an \emph{$x$-partition} of $E$.
\end{definition}

\begin{definition}
  For a polynomial matrix \(\mat{A} \in \xmatspace{2\mu}{\mu}\) and
  polynomials \(u_0,\ldots,u_{\mu-1} \in \xring\setminus\{0\}\), we define
    \(\H_{\vec{u}}(\mat{A})\)
    as
  \[
    \Big\{ \vec{v} \in \xmatspace{1}{2\mu} \mid \vec{v} \mat{A}_{*,k} = 0 \bmod u_k \quad\text{for } 0\le k < \mu\Big\},
  \]
  where \(\mat{A}_{*,k}\) is the column \(k\) of \(\mat{A}\).
\end{definition}

Note that we have the following inclusion of \(\xring\)-submodules:
\[
  (\textstyle\prod_{0\le k < \mu} u_k) \xmatspace{1}{2\mu}
  \;\;\subseteq\;\;
  \H_{\vec{u}}(\mat{A})
  \;\;\subseteq\;\;
  \xmatspace{1}{2\mu}.
\]
In particular, \(\H_{\vec{u}}(\mat{A})\) is a free \(\xring\)-module of rank
\(2\mu\), and each of its bases can be represented as a nonsingular
\(2\mu\times2\mu\) matrix over \(\xring\).

\begin{lemma}
  \label{lem:fxbasis_pade}
  Let $a \in \Ya(A)$. Let $U_0,\dots,U_{\mu-1}$ be an $x$-partition of $E$, and let
  $\mat{S} = [S_{i,k}]$ and $\mat{T} = [T_{i,k}]$ be matrices in
  $\xmatspace{\mu}{\mu}$ such that
  \begin{align*}
    S_{i,k}(x(E_j)) = \y[A+B]_i(E_j)
  \end{align*}
  and
  \[
    T_{i,k}(x(E_j)) = a(E_j)\y[B]_i(E_j)
  \]
    for \(E_j \in U_k\).
  If $\vec{u} = (u_0,\dots,u_{\mu-1}) \in \xring^\mu$, where $u_k = \prod_{E_j \in
  \supp(U_k)}(x - x(E_j))$, then the map
  \[
    \psi: \sum_{i = 0}^{\mu-1}(s_i \y[A+B]_i + t_i z \y[B]_i) \mapsto (s_0,\dots,s_{\mu-1},t_0,\dots,t_{\mu-1})
  \]
  is an $\xring$-isomorphism between $\N_{A,B,E}(a)$ and
  $\H_{\vec{u}}(\mat{A})$, where
  \[
    \mat{A} =
    \begin{bmatrix}
      \mat{S} \\
      \mat{T}
    \end{bmatrix}
    \in \xmatspace{2\mu}{\mu}.
  \]
\end{lemma}
\begin{proof}
  Clearly $\psi$ is an $\xring$-isomorphism between $\Ya(A+B) \oplus z \Ya(B)$
  and $\xring^{2\mu}$, therefore it suffices to show that for any $Q \in
  \Ya(A+B) \oplus z \Ya(B)$ it holds that $Q \in \N_{A,B,E}(a)$ if and only if
  $\psi(Q) \in \H_{\vec{u}}(\mat{A})$, i.e.\ that for all $k=0,\dots,\mu-1$,
  $Q(E_j,a(E_j)) = 0$ for all $E_j \in \supp(U_k)$ if and only if $\psi(Q)
  \cdot \mat{A}_{*,k} = 0 \bmod u_k$. This is true since for every $E_j \in U_k$ the following
  identity holds, where $\alpha = x(E_j)$:
  \begin{align*}
    & Q(E_j, a(E_j)) \\
    &= \sum_{i=0}^{\mu-1}\big( s_i(\alpha) \y[A+B]_i(E_j) + a(E_j)t_i(\alpha)\y[B]_i(E_j) \big ) \\
    &= \sum_{i=0}^{\mu-1}\big( s_i (\alpha) S_{i,k}(\alpha) + t_i(\alpha)T_{i,k}(\alpha) \big) \\
    &= (\psi(Q) \cdot \mat{A}_{*,k})(\alpha). \qedhere
  \end{align*}
\end{proof}

\begin{lemma}
  \label{lem:popov-submatrix}
  Keeping notation as in \cref{lem:fxbasis_pade}, let $\mat{P} \in
  \xmatspace{(2\mu)}{(2\mu)}$ be the $\vec{d}$-Popov basis of
  $\H_{\vec{u}}(\mat{A}) = \psi(\N_{A,B,E}(a))$, where $\deg(E) \geq 2g  + \mu
  + \delta_A(a) + \deg(A)$ and
  \(\vec{d} = \textstyle\frac{1}{\mu}(\vec{e} + (\deg(B),\ldots,\deg(B))) \in \textstyle\frac{1}{\mu}\ZZ^{2\mu}\) with
  \begin{align*}
    \vec{e} =  \big( &\delta_{A+B}(\y[A+B]_0), \dots, \delta_{A+B}(\y[A+B]_{\mu-1}),\\
    &  \delta_B(\y[B]_0) + \delta_A(a), 
  \dots, \delta_B(\y[B]_{\mu-1}) + \delta_A(a)\big) \in \ZZ^{2\mu}.
  \end{align*}
  Then exactly $\mu$ rows of $\mat{P}$ have $\vec{d}$-degree less than
  $\frac{1}{\mu}(\deg(E) - \deg(A))$. Furthermore, if $\tilde{\mat{P}} \in
  \xmatspace{\mu}{(2\mu)}$ is the submatrix of $\mat{P}$ consisting of these
  rows, then for $k = 0,\dots,\mu-1$ the row \(k\) of $\tilde{\mat{P}}$ is
  $\psi(Y_k)$, where
  \[
    Y_k = -a\y[B]_{k} +z\y[B]_{k} \in (z - a)\Ya(B) \subset \N_{A,B,E}(a).
  \]
  Consequently, if $\tilde{\mat{P}} = [\mat{P}_1 \;\; \mat{P}_2]$, where
  $\mat{P}_1$ and $\mat{P}_2$ are in $\xring^{\mu \times \mu}$, then
  $a\y[B]_{k} = \sum_{i=0}^{\mu-1} p_{k,i} \y[A+B]_i$, where
  $(p_{k,0},\dots,p_{k,\mu-1})$ is the row \(k\) of $-\mat{P}_1$.
\end{lemma}
\begin{proof}
  We start with some observations on the matrix whose rows are $\psi(Y_k)$ for
  $k = 0, \ldots, \mu-1$. This is a \(\mu \times (2\mu)\) matrix over
  \(\xring\), whose rank is \(\mu\) since $Y_0,\dots,Y_{\mu-1}$ are
  \(\xring\)-linearly independent. By construction, its \(\mu \times \mu\)
  rightmost submatrix is the identity matrix. Writing $Y_k =
  -\sum_{i=0}^{\mu-1} w_i\y[A+B]_i + z \y[B]_{k}$, where $w_i \in \xring$, the
  fact that $Y_k(a) = 0$ implies
  \begin{align*}
    \max_i \delta_{A+B}(w_i\y[A+B]_i) & = \delta_{A+B}\left(\sum_{i=0}^{\mu-1} w_i\y[A+B]_i\right) \\
                                      & =  \delta_{A+B}(a \y[B]_{k})
                                      = \delta_B(\y[B]_{k}) + \delta_A(a).
  \end{align*}
  Consequently, $\deg_{\vec{d}}(\psi(Y_k)) = \frac{1}{\mu}(\delta_B(\y[B]_{k})
  + \delta_A(a)+\deg(B))$, and this \(\vec{d}\)-degree is reached at index
  \(\mu+k\). This shows that $\mu + k$ is the $\vec{d}$-pivot index of the row
  $\psi(Y_k)$. This property combined with the special shape (with an identity
  submatrix) of the matrix formed by the $\psi(Y_k)$'s ensure that this matrix
  is in \(\vec{d}\)-Popov form.

  Furthermore, since for $k = 0,\dots,\mu-1$, \(\psi(Y_k)\) is in
  $\H_{\vec{u}}(\mat{A})$ with
  \[
    \deg_{\vec{d}}(\psi(Y_k)) < \frac{1}{\mu}(\delta_A(a) + 2g  + \mu) \leq \frac{1}{\mu}(\deg(E) - \deg(A)),
  \]
  where the strict inequality is due to
  \cref{lem:Fx_basis}, then at least $\mu$ rows of $\mat{P}$ have
  $\vec{d}$-degree less than $\frac{1}{\mu}(\deg(E) - \deg(A))$, because
  $\mat{P}$ is $\vec{d}$-row reduced.

  Now, for any $Q = Q_0 + z Q_1 \in \N_{A,B,E}(a),$ where $Q_0 =
  \sum_{i=0}^{\mu-1} s_i\y[A+B]_i \in \Ya(A+B)$ and $Q_1 = \sum_{i=0}^{\mu-1}
  t_i\y[B]_i \in \Ya(B)$ with $s_i,t_i \in \xring$, it holds that
  \(\deg_{\vec{d}}(\psi(Q))\) is equal to
  \begin{align*}
    &\max\Bigg\{
      \max_i \left( \deg(s_i) + \frac{\delta_{A+B}(\y[A+B]_i) + \deg(B)}{\mu} \right), \\
    & \hphantom{somespac} \max_i\left( \deg(t_i) + \frac{\delta_B(\y[B]_i) + \delta_A(a) + \deg(B)}{\mu} \right)
      \Bigg\} \\
    &= \frac{1}{\mu}\big(\max\{ \delta_{A+B}(Q_0), \delta_B(Q_1) + \delta_A(a)\}  + \deg(B)\big).
  \end{align*}
  It then follows from \cref{lem:fxbasis_basismem} that
  \[
    \deg_{\vec{d}}(\psi(Q)) < \frac{1}{\mu}(\deg(E) - \deg(A)) \implies Q \in (z-a)\Ya(B),
  \]
  which means that at most $\mu$ rows of $\mat{P}$ can have $\vec{d}$-degree
  less than $\frac{1}{\mu}(\deg(E) - \deg(A))$, because $(z-a)\Ya(B)$ has rank
  $\mu$ as an $\xring$-module.

  Thus, exactly $\mu$ rows of $\mat{P}$ have $\vec{d}$-degree less than
  $\frac{1}{\mu}(\deg(E) - \deg(A))$, which proves the first claim of the
  lemma. For the second claim, the above observations show that the matrix
  formed by the $\psi(Y_k)$'s is a left multiple \(\mat{U}\tilde{\mat{P}}\) by
  a nonsingular \(\mu \times \mu\) matrix \(\mat{U}\). Yet, since both
  \(\tilde{\mat{P}}\) and the matrix whose rows are the $\psi(Y_k)$'s are in
  \(\vec{d}\)-Popov form, and since the rightmost \(\mu\times\mu\) submatrix of
  the latter is the identity matrix, the only possibility is \(\mat{U} =
  \idmat_{\mu}\), proving the second claim. The last claim is obvious.
\end{proof}

\begin{algorithm}[htp]
  \algoCaptionLabel{BasisProducts}{A,B,E,\vec{a},\vec{x},\vec{y}^{(A+B)},\vec{y}^{(B)}}
  \begin{algorithmic}[1]

    \Require ~ 
      \begin{itemize}[leftmargin=0cm,itemsep=0pt]
        \item divisors $A$ and $B$,
        \item a divisor $E = E_1 + \cdots + E_N$,
          where $E_1,\dots,E_N$ are pairwise distinct rational places such that
          $\supp(E) \cap (\supp(A) \cup \supp(B) \cup \{P_\infty\}) = \emptyset$
        \item evaluations \(\vec{a} = (a_j)_{j=1,\dots,N}\), where \(a_j = a(E_j)\)
          for a function $a \in \Ya(A)$ with known \(\delta_A(a)\) and such that $\deg(E) \geq \deg(A) + \delta_A(a) + 2g  + \mu$,
        \item evaluations $\vec{x} = (x_j)_{j=1,\dots,N}$, where $x_j = x(E_j) \in \field$,
        \item evaluations $\vec{y}^{(A+B)} = (\y[A+B]_{i,j})^{i=0,\dots,\mu-1}_{j=1,\dots,N}$, where $\y[A+B]_{i,j} = \y[A+B]_i(E_j) \in \field$,
        \item evaluations $\vec{y}^{(B)} = (\y[B]_{i,j})^{i=0,\dots,\mu-1}_{j=1,\dots,N}$, where $\y[B]_{i,j} = \y[B]_i(E_j) \in \field$.
    \end{itemize}

    \Ensure matrix \([p_{k,i}] \in \xmatspace{\mu}{\mu}\) of the
    $\xring$-linear map $f \in \Ya(B) \mapsto af \in \Ya(A+B)$
    with respect to the ordered $\xring$-bases $(\y[B]_0,\dots,\y[B]_{\mu-1})$
    for $\Ya(B)$ and $(\y[A+B]_0,\dots,\y[A+B]_{\mu-1})$ for $\Ya(A+B)$, meaning
     $a\y[B]_{k} = \sum_{i=0}^{\mu-1} p_{k,i} \y[A+B]_i$ for all
    \(k \in \{0,\ldots,\mu-1\}\).

    \State\InlineIf{$\vec{a}=\vec{0}$}{\Return matrix $\zmat \in \xmatspace{\mu}{\mu}$}
    \State $U_0,\dots,U_{\mu-1} \gets $ an $x$-partition of $E$
    \State $\mat{S} = [S_{i,k}] \in \xring^{\mu \times \mu} \gets$
    matrix with $S_{i,k}(x_j) = \y[A+B]_{i,j}$ for all \(j\) such that $E_j \in U_k$
    \State $\mat{T} = [T_{i,k}] \in \xring^{\mu \times \mu} \gets$
    matrix with $T_{i,k}(x_j) = a_j \, \y[B]_{i,j}$ for all \(j\) such that $E_j \in U_k$
    \State $\vec{u} = [u_k] \in \xring^{\mu} \gets$
    vector with $u_k = \prod_{j : E_j \in U_k}(x - x_j)$
    \State $\vec{d} \in \frac{1}{\mu}\ZZ^{2\mu} \gets \frac{1}{\mu} (\vec{e} + (\deg(B), \ldots, \deg(B)))$, where
    \Statex \(~~~~~\vec{e} = \big( \delta_{A+B}(\y[A+B]_0),\ldots,\delta_{A+B}(y^{(A+B)}_{\mu-1}),\)
    \Statex \(~~~~~~~~\delta_B(\y[B]_0) + \delta_A(a), \ldots, \delta_B(\y[B]_{\mu-1}) + \delta_A(a)
      \big) \in \ZZ^{2\mu}$
    \State $\mat{P} \in \xmatspace{(2\mu)}{(2\mu)} \gets$ $\vec{d}$-Popov basis of $\H_{\vec{u}}(\mat{A})$ where
    $\mat{A} =
    [\begin{smallmatrix}
        \mat{S} \\
        \mat{T}
      \end{smallmatrix}]
      \in \xmatspace{(2\mu)}{\mu}$
    \State $[\mat{P}_1 \;\; \mat{P}_2] \in \xmatspace{\mu}{(2\mu)} \gets$ the submatrix of $\mat{P}$ consisting of all rows with
    $\vec{d}$-degree less than $\frac{1}{\mu}(\deg(E) - \deg(A))$, where $\mat{P}_1, \mat{P}_2 \in \xring^{\mu \times \mu}$
    \State \Return \(-\mat{P}_1\)
  \end{algorithmic}
\end{algorithm}

\begin{theorem}\label{thm:basis-products}
  \cref{algo:BasisProducts} is correct and costs $\softO(\mu^{\expmm-1}(N+|\deg(A)|))$ operations in $\field$.
\end{theorem}
\begin{proof}
  Correctness is given by \cref{lem:popov-submatrix}. For complexity, simply
  note that the computational bottleneck lies in Step 8, in which case
  $\delta_A(a) \ge -\deg(A)$ because $a$ is nonzero and $a \in
  \L(\delta_A(a)\Pinf + A)$. By assumption, we have that $N=\deg(E) \geq
  \deg(A) + \delta_A(a) + 2g  + \mu$, hence by \cref{lem:Fx_basis}
  \begin{align*}
    -\deg(A) & \le \delta_{A+B}(\y[A+B]_i)+\deg(B) \\
             & \le 2g-1-\deg(A) +\mu \\
             & < \deg(E)- 2\deg(A) -\delta_A(a)\\
             & \le \deg(E)-\deg(A)=N-\deg(A)
  \end{align*}
  and
  \begin{align*}
    -\deg(A) &\le \delta_A(a) \le \delta_{B}(\y[B]_i)+\delta_A(a)+\deg(B)\\
             & \le 2g-1 +\mu + \delta_A(a)\\
             & \le -1 + \deg(E) - \deg(A) < N-\deg(A).
  \end{align*}

  Since $\deg(u_k) \le N/\mu$ for $k=0,\dots,\mu-1$, then the total complexity of
  the algorithm is given by \cite[Cor.\,V.10]{BRS2022} as
  \begin{align*}
    & \softO \big( \mu^{\expmm-1}\max\{|\deg(E)|,|\deg(E)-\deg(A)|,|\deg(A)|\} \big) \\
    & \subseteq \softO(\mu^{\expmm-1}(N+|\deg(A)|))
  \end{align*}
  operations in $\field$.
\end{proof}

\subsection{Computing a small-degree \texorpdfstring{\(\xring\)}{Fq[x]}-basis of \texorpdfstring{\(\M\)}{M[s,l](D,G)}}
\label{sec:algorithms:construct_Fqx_basis}

\begin{algorithm*}[htp]
  \algoCaptionLabel{InterpolantPolMatBasis}{\vec{r},D,G,E,\vec{x},\vec{y}^{(-tG)},\vec{y}^{(G_t)},\vec{y}^{(H_t)}}
  \begin{algorithmic}[1]

    \Require ~ 
    \begin{itemize}[leftmargin=0cm,itemsep={0pt}]
      \item received word \(\vec{r} \in \field^n\),
      \item the code divisors \(D\) and \(G\),
      \item a divisor $E = E_1 + \cdots + E_N$, where $E_1,\dots,E_N$ are pairwise distinct rational places
        not in $\{P_{\infty}\} \cup \supp(G)$, with $\deg(E) \geq sn + 4g + \mu - 1$,
      \item evaluations $\vec{x} = (x_j)_{j=1,\dots,N}$, where $x_j = x(E_j) \in \field$,
      \item evaluations $\vec{y}^{(-tG)} = (\y[-tG]_{i,j})^{i=0,\dots,\mu-1}_{j=1,\dots,N}$ for \(t=-1,0,\ldots,\ell\),
        where $\y[-tG]_{i,j} = \y[-tG]_i(E_j) \in \field$,
      \item evaluations $\vec{y}^{(G_t)} = (\y[G_t]_{i,j})^{i=0,\dots,\mu-1}_{j=1,\dots,N}$ for \(t=0,\ldots,s-1\),
        where \(G_t = (t-s)D - tG\) and $\y[G_t]_{i,j} = \y[G_t]_i(E_j) \in \field$,
      \item evaluations $\vec{y}^{(H_t)} = (\y[H_t]_{i,j})^{i=0,\dots,\mu-1}_{j=1,\dots,N}$ for \(t=0,\ldots,s-1\),
        where \(H_t = -sD - tG\) and $\y[H_t]_{i,j} = \y[H_t]_i(E_j) \in \field$.
    \end{itemize}

    \Ensure a matrix \(\smatM := [\begin{smallmatrix} \mat{D} & \zmat \\
    \mat{\bar{R}} & \idmat \end{smallmatrix}]\in \xmatspace{\mdim}{\mdim}\) as
    in \cref{thm:general_small_basis}: \(\smatM\) is a basis of \(\M\) seen as
    an \(\xring\)-submodule of \(\bigoplus_{0 \le t \le \ell} z^t \Ya(-tG)\)
    and \(\deg(\mat{\bar{R}})\) has degree in \(\bigO(s(n+g)/\mu)\).

    \State \CommentLine{Compute matrices \(\mat{D}_t\) and \(\mat{E}_t\) in \(\xmatspace{\mu}{\mu}\), see \cref{def:inclusion_map,def:inclusion_map_superset}}
      \label{step:matrix:DtEt_start}
    \For{\(t=0,\ldots,s-1\)}
      \State \(\mat{D}_t \gets \Call{algo:BasisProducts}{(s-t)D,G_t,E,(1,\ldots,1),\vec{x},\vec{y}^{(-tG)},\vec{y}^{(G_t)}}\)
        \label{step:matrix:DtEt_Dt}
      \State \(\mat{E}_t \gets \Call{algo:BasisProducts}{sD,H_t,E,(1,\ldots,1),\vec{x},\vec{y}^{(-tG)},\vec{y}^{(H_t)}}\)
        \label{step:matrix:DtEt_end}
    \EndFor

    \State \CommentLine{Compute matrices \(\mat{R}_1,\ldots,\mat{R}_{\ell}\) in \(\xmatspace{\mu}{\mu}\), see \cref{def:R_t}}
      \label{step:matrix:Rt_start}
    \State \(R \in \Ya(G) \gets \hyperlink{cite.BRS2022}{\textproc{Interpolate}}(\vec{r},D,G,\vec{x},\vec{y}^{(G)})\)
          \Comment{\cite[Algorithm~2]{BRS2022}}
          \label{step:matrix:Interpolate}
    \State \(\vec{\hat{r}} \in \field^N \gets \hyperlink{cite.BRS2022}{\textproc{Evaluate}}(R,E,G,\vec{x},\vec{y}^{(G)})\)
          \Comment{\cite[Algorithm~1]{BRS2022}}
          \label{step:matrix:Evaluate}
    \For{\(t=1,\ldots,\ell\)}
    \State \(\mat{R}_t \gets \Call{algo:BasisProducts}{G,-tG,E,\vec{\hat{r}},\vec{x},\vec{y}^{(-(t-1)G)},\vec{y}^{(-tG)}}\)
      \label{step:matrix:Rt_end}
    \EndFor

    \State \CommentLine{Compute matrix \(\mat{D} \in \xmatspace{(\mu s)}{(\mu s)}\), see \cref{def:Fqx_basis_matrix_of_M}}
          \label{step:matrix:D_start}
    \State \(\mat{D} = [\mat{D}^{(t,j)}]_{0\le t < s, 0 \le j < s}
    \gets \operatorname{Diag}(\mat{D}_0,\ldots,\mat{D}_{s-1})\),
    where \(\mat{D}^{(t,j)} \in \xmatspace{\mu}{\mu}\)
    \For{\(t=1,\ldots,s-1\)}
      \State \InlineFor{\(j=t-1,\ldots,0\)}{\(\mat{D}^{(t,j)} \gets \mat{D}^{(t,j+1)} \mat{R}_{j+1}\)}
          \label{step:matrix:DR_product}
      \State \InlineFor{\(j=t-1,\ldots,0\)}{\(\mat{D}^{(t,j)} \gets (-1)^{t-j} \binom{t}{j} \mat{D}^{(t,j)}\)}
          \label{step:matrix:D_end}
    \EndFor

    \State \CommentLine{Compute matrix \(\mat{\bar{R}} \in \xmatspace{((\ell+1-s)\mu)}{(\mu s)}\), see \cref{def:Fqx_basis_matrix_of_M,thm:specific_small_basis}}%
          \label{step:matrix:R_start}
    \State \(\mat{\bar{R}} = [\mat{\bar{R}}^{(t,j)}]_{s\le t \le \ell, 0 \le j < s} \gets\) zero matrix
    with blocks \(\mat{\bar{R}}^{(t,j)} \in \xmatspace{\mu}{\mu}\)

    \For{\(j=0,\ldots,s-1\)}
      \State \(\mat{\bar{R}}^{(s,j)} \gets \hyperlink{href_lem_matrix_quorem}{\textproc{PM-Rem}}(\mat{R}_s\mat{R}_{s-1} \cdots \mat{R}_{j+1}, \mat{E}_j)\) \Comment{algorithm from \cref{lem:matrix_quorem}}
          \label{step:matrix:R_reminit}
      \State \InlineFor{\(t=s+1,\ldots,\ell\)}{\(\mat{\bar{R}}^{(t,j)} \gets \hyperlink{href_lem_matrix_quorem}{\textproc{PM-Rem}}(\mat{R}_t \mat{\bar{R}}^{(t-1,j)}, \mat{E}_j)\)}
          \label{step:matrix:R_remiter}
      \State \InlineFor{\(t=s,\ldots,\ell\)}{\(\mat{\bar{R}}^{(t,j)} \gets \gamma_{t,j} \mat{\bar{R}}^{(t,j)}\) \Comment{\(\gamma_{t,j}\) defined in \cref{eqn:dfn_gamma}}}
          \label{step:matrix:R_end}
    \EndFor
    \State \Return \([\begin{smallmatrix} \mat{D} & \zmat \\ \mat{\bar{R}} & \idmat \end{smallmatrix}] \in \xmatspace{((\ell+1)\mu)}{((\ell+1)\mu)}\)
  \end{algorithmic}
\end{algorithm*}

\begin{theorem}
  \label{thm:algo_construct_Fqx_basis}
  \cref{algo:InterpolantPolMatBasis} is correct and costs \(\softO(\ell
  \mu^{\expmm-1} N + s^2 \ell \mu^{\expmm-1} (n+g))\) operations in \(\field\).
\end{theorem}
\begin{proof}
  The correctness follows from the correctness of the called algorithms and
  from the results in \cref{sec:interpolant_module:small_Fqx_basis}.

  For complexity, let us first consider the total cost of the calls to
  \Call{algo:BasisProducts}{}; for completeness we also give at the same time
  the detailed verification that the constraint on \(\deg(E)\) required in the
  input of this algorithm is satisfied.
  \begin{itemize}
    \item
      For \(t=0,\ldots,s-1\), the call at \cref{step:matrix:DtEt_Dt} is for the
      divisor \(A = (s-t)D\) and the function \(a = 1\), and therefore costs
      \(\softO(\mu^{\expmm-1} (N+|\deg((s-t)D)|)) = \softO(\mu^{\expmm-1}
      (N+(s-t)n))\) according to \cref{thm:basis-products}. Over the \(s\)
      iterations, this cost is in \(\softO(s \mu^{\expmm-1} N + s^2
      \mu^{\expmm-1} n)\). Furthermore the input requirements of
      \Call{algo:BasisProducts}{} impose \(\deg(E) \ge \deg((s-t)D) +
      \delta_{(s-t)D}(1) + 2g+\mu\) for all \(t=0,\ldots,s-1\), hence we must
      ensure \(\deg(E) \ge sn + 2g + \mu\); this is implied by the input
      requirements of \Call{algo:InterpolantPolMatBasis}{}.
    \item
      The calls at \cref{step:matrix:DtEt_end} are for the divisor \(A =
      (s-t)D\) and the function \(a = 1\). Thus their total complexity fits
      within the one in the previous item, and these calls do not bring any
      additional restriction on \(\deg(E)\).
    \item
      Finally, the calls to \Call{algo:BasisProducts}{} at
      \cref{step:matrix:Rt_end} are for the divisor \(A = G\) and the function
      \(a = R\), for each of the \(\ell\) iterations. In total, this costs
      \(\softO(\ell \mu^{\expmm-1} (N+|\deg(G)|)) = \softO(\ell \mu^{\expmm-1}
      N + \ell \mu^{\expmm-1} (n+g))\). These calls all add the same constraint
      on \(\deg(E)\), namely \(\deg(E) \ge \deg(G) + \delta_G(R) + 2g + \mu\).
      Since \(\delta_G(R) \le n + 2g - 1 - \deg(G)\) holds by construction of
      \(R\) (see the output specification of \cite[Algorithm~2]{BRS2022}), this
      constraint is satisfied when \(\deg(E) \ge n + 4g + \mu - 1\), and this
      inequality is indeed implied by the input requirements of
      \Call{algo:InterpolantPolMatBasis}{}.
  \end{itemize}

  The interpolation at \cref{step:matrix:Interpolate} costs
  \(\softO(\mu^{\expmm-1} (N+g))\), by \cite[Lemma~V.12]{BRS2022}. The
  evaluation at \cref{step:matrix:Evaluate} costs \(\softO(\mu N + \delta_G(R)
  + \deg(G))\), by \cite[Lemma~V.2]{BRS2022}; our assumption on \(\deg(E)\)
  implies that this is in \(\softO(\mu N)\).

  The costly part of the computation of \(\mat{D}\) at
  \crefrange{step:matrix:D_start}{step:matrix:D_end} is the matrix products at
  \cref{step:matrix:DR_product}. For each \(t=1,\ldots,s-1\), we start from
  \(\mat{D}_t\) which has degree in \(\bigO(s(n+g)/\mu)\) (see
  \cref{lem:deg_D_t}), and then we multiply iteratively for \(j=t-1,\ldots,0\)
  by \(\mat{R}_j\) whose degree is in \(\bigO((n+g)/\mu)\) (see
  \cref{lem:deg_R_t}). Thus, altogether we perform about \(\frac{s^2}{2}\)
  multiplications of two \(\mu\times\mu\) matrices of degree in
  \(\bigO(s(n+g)/\mu)\), for a total cost of \(\softO(s^3 \mu^{\expmm-1}
  (n+g))\).

  The costly part of the computation of \(\mat{\bar{R}}\) at
  \crefrange{step:matrix:R_start}{step:matrix:R_end} is the matrix products and
  matrix remainders at both \cref{step:matrix:R_reminit,step:matrix:R_remiter}.
  Consider a fixed iteration \(j\), for some \(j \in \{0,\ldots,s-1\}\). The
  above recalled bound on \(\deg(\mat{R}_t)\) ensures that the product
  \(\mat{R}_s\mat{R}_{s-1} \cdots \mat{R}_{j+1}\) at
  \cref{step:matrix:R_reminit} can be computed in \(\softO(s^2 \mu^{\expmm-1}
  (n+g))\), and has degree in \(\bigO(s(n+g)/\mu)\). Then, since
  \(\deg(\mat{E}_j)\) is in \(\bigO(s(n+g)/\mu)\) as well (see
  \cref{lem:deg_E_t}), the matrix division with remainder at the same line
  costs \(\softO(s \mu^{\expmm-1} (n+g))\) and returns a matrix whose degree is
  in \(\bigO(s(n+g)/\mu)\) (see \cref{lem:matrix_quorem}). At
  \cref{step:matrix:R_remiter} there are \(\le \ell\) iterations, and similarly
  to \cref{step:matrix:R_reminit}, each of them performs a matrix product and
  then a matrix division with remainder which both cost \(\softO(s
  \mu^{\expmm-1} (n+g))\), for a total of \(\softO(s \ell \mu^{\expmm-1}
  (n+g))\). Note that degrees remain controlled since each of these iterations
  produces a matrix remainder whose degree is less than \(\deg(\mat{E}_j)\),
  which is in \(\bigO(s(n+g)/\mu)\). Summing over the iterations for \(j =
  0,\ldots,s-1\), and using \(s\le \ell\), we get a cost bound of \(\softO(s^2
  \ell \mu^{\expmm-1} (n+g))\) operations in \(\field\) for
  \crefrange{step:matrix:R_start}{step:matrix:R_end}.

  Finally, summing the costs of each analyzed part above
  yields the result.
\end{proof}

\section{Decoder with better complexity}
\label{sec:decoder}

\subsection{The decoding algorithm}
\label{sec:decoder:decoding}

The overall decoding algorithm is the one in \cite[Algorithm~7]{BRS2022} with
the first three lines replaced by a call to
\Call{algo:InterpolantPolMatBasis}{}, which provides \(\pfM\) in complexity
$\softO(s^2 \ell \mu^{\expmm-1}(n+g))$ according to
\cref{thm:algo_construct_Fqx_basis}; indeed one can take \(N \in \bigO(sn + g +
\mu) \subset \bigO(s(n + g))\). After that, two expensive computations remain.
The first one asks to find the shifted Popov form of \(\pfM\)
\cite[Algorithm~7, Line~5]{BRS2022}, which costs
$\softO(s^2\ell^{\expmm-1}\mu^{\expmm-1}(n+g) + \ell^\expmm \mu^\expmm)$ operations in \(\field\)
according to \cref{thm:general_small_basis}. The second one is the root finding
step \cite[Algorithm~7, Line~10]{BRS2022}, whose complexity is in
$\softO(s\ell\mu^{\expmm-1}(n+g))$ as detailed in
\cref{sec:decoder:rootfinding}. Hence the overall cost bound
$\softO(s^2\ell^{\expmm-1}\mu^{\expmm-1}(n+g) + \ell^\expmm \mu^\expmm)$ for the list decoder.

Due to the modification of the first steps, the precomputed data slightly
differs from the one listed in \cite[Section~VI]{BRS2022}. Here, we do not need
to know the evaluations of \(\Ya\)-module generators for \(\Ya(G_t)\), \(t =
0,\ldots,\ell\), denoted by ``\(\vec{g}\)'' in the above reference. As a kind
of replacement, we need the evaluations $\vec{y}^{(G_t)}$ and $\vec{y}^{(H_t)}$
for \(t=0,\ldots,s-1\) and $\vec{y}^{(-tG)}$ for \(t=0,\ldots,\ell\), as
defined in the input of \Call{algo:InterpolantPolMatBasis}{}. Observe that this
algorithm also requires $\vec{y}^{(-tG)}$ for \(t=-1\); but \(\vec{y}^{(G)}\)
is already part of the precomputation in \cite[Section~VI]{BRS2022} (denoted by
\(\vec{y}\)). Except for ``\(\vec{g}\)'', the rest of the precomputed data
listed in \cite[Section~VI]{BRS2022} is kept as such.

\subsection{The root finding step}
\label{sec:decoder:rootfinding}

In \cite{BRS2022}, an algorithm is given that finds all roots of the found
polynomial $Q(z) \in \M$ in complexity $\softO(\ell^2\mu^{\expmm-1}(n+g))$. The
term $\ell^2$ is at odds with our target complexity. Fortunately, a slightly
better complexity analysis shows that \cite[Algorithm~6]{BRS2022} actually has
complexity $\softO(s\ell\mu^{\expmm-1}(n+g))$. More precisely, in the proof of
\cite[Proposition V.33]{BRS2022}, the $\ell^2$ term comes from the estimates
$\softO(\mu \ell\beta) \subseteq \softO(\ell^2\mu(n+g))$ and $\softO(\beta
\deg_z(\hat{Q})) \subseteq \softO(\ell^2(n+g))$, where $\beta$ is chosen such
that $\beta \ge 2 \ell \deg(G) + s (n-\tau)$ and where $\deg_z(\hat{Q})=\ell$.
In \cite{BRS2022} the estimate $\deg(G) \in \bigO(n+g)$ is used to show the
mentioned inclusions. A third part of the complexity analysis of
\cite[Algorithm~6]{BRS2022} adds a term $\softO(\ell \mu^{\expmm-1}(n+g))$,
yielding as total complexity the mentioned $\softO(\ell^2\mu^{\expmm-1}(n+g))$.

However, the root finding has as input a polynomial $Q \in \M$ satisfying
$\delta_G(Q) < s(n-\tau)$. In particular $\delta_{-\ell G}(Q_\ell) <
s(n-\tau)$, which implies that $Q_\ell \in \L(-\ell G+s \Pinf)$. This implies
that either $-\ell \deg(G) +sn \ge 0$, or $Q_\ell=0$ in all cases. In the
latter case one might as well have started the decoding algorithm for a smaller
value of designed list size $\ell$. We may conclude that without loss of
generality one can assume $\ell \deg(G) \le sn$. This implies that $\beta \in
\bigO(sn)$ and therefore $\softO(\mu \ell\beta) \subseteq \softO(s\ell\mu n)$
and $\softO(\beta \deg_z(\hat{Q})) \subseteq \softO(s\ell n)$. Leaving the
remaining part of the complexity analysis exactly the same as in the proof of
\cite[Proposition~V.33]{BRS2022}, we see that the root finding part can be
handled using \cite[Algorithm~6]{BRS2022} in complexity
$\softO(s\ell\mu^{\expmm-1}(n+g))$.

\section*{Acknowledgments}

The first author would like to acknowledge the support from The Danish Council for Independent Research (DFF-FNU) for the project \emph{Correcting on a Curve}, Grant No.~8021-00030B.
The second author would like to acknowledge the support from Sorbonne Universit\'e's Faculty of Science and Engineering through the project \emph{Tremplin 2022: Fast reconstruction of multivariate algebraic relations}; from the Agence nationale de la recherche (ANR), grant agreement ANR-23-CE48-0003-01 CREAM; from the ANR\&Austrian Science Fund FWF, grant agreement ANR-22-CE91-0007 EAGLES.


\begin{IEEEbiographynophoto}{Peter Beelen}
  received the master’s degree in mathematics from the University of Utrecht,
  The Netherlands, in 1996, and the Ph.D. degree in mathematics from the
  Technical University of Eindhoven, The Netherlands, in 2001. Since October
  2004, he has been a Staff Member with the Technical University of Denmark
  (DTU), Kongens Lyngby, Denmark. He has been an Assistant Professor with DTU,
  until January 2007, and an Associate Professor, until August 2014. Since
  September 2014, he has been with DTU as a Professor. His research interests
  include various aspects of algebra and its applications, notably algebraic
  curves, and algebraic coding theory.
\end{IEEEbiographynophoto}

\begin{IEEEbiographynophoto}{Vincent Neiger}
  received the master’s degree in Computer Science from the \'Ecole Normale
  Sup\'erieure de Lyon, France, in 2012, and the Ph.D. degree in Computer
  Science from the Universit\'e de Lyon, France, in 2016. From 2017 to 2021, he
  was a \emph{Ma\^itre de Conf\'erences} at University of Limoges, France.
  Since 2021, he has been a \emph{Ma\^itre de Conf\'erences} at Sorbonne
  Universit\'e, France. His research interests revolve around algebraic
  computations and their applications, with a focus on the design and
  implementation of efficient algorithms for polynomial and matrix
  computations.
\end{IEEEbiographynophoto}

\end{document}